\RequirePackage{fix-cm}
\documentclass[smallextended]{svjour3}       
\smartqed  

\usepackage{graphicx}

\usepackage[utf8]{inputenc}
\usepackage{amsfonts}
\usepackage{amssymb}
\usepackage{natbib}
\usepackage{mathtools}
\usepackage{dsfont}
\usepackage{makeidx}
\usepackage{epstopdf}
\usepackage{float}
\usepackage{booktabs}
\usepackage{ragged2e}
\usepackage{url}
\usepackage{csquotes}
\usepackage{color}
\newcommand{\blue}{\textcolor{black}}

\DeclareMathOperator*{\argmax}{arg\,max}
\newcommand{\RR}{{\mathbb R}}

 \journalname{Journal of Mathematical Biology}

\usepackage{enumitem}

\begin{document}
\title{Combinatorial properties of phylogenetic diversity indices}

\author{Kristina Wicke       \and
				Mike Steel
}


\institute{K. Wicke \at
              Institute of Mathematics and Computer Science, University of Greifswald, Germany. Orchid ID: 0000-0002-4275-5546
           \and
     M. Steel \at				
					Biomathematics Research Centre, University of Canterbury, Christchurch, New Zealand (corresponding author). Orchid ID: 0000-0001-7015-4644\\
              \email{mike.steel@canterbury.ac.nz}
}

\date{Received: date / Accepted: date}

\maketitle

\begin{abstract}
Phylogenetic diversity indices provide a formal way  to apportion ‘evolutionary heritage’  across species.  Two natural diversity indices are Fair Proportion (FP) and Equal Splits (ES).  FP is also called `evolutionary distinctiveness' and, for rooted trees, is identical to the Shapley Value (SV), which arises from cooperative game theory. In this paper, we investigate the extent to which FP and ES can differ, characterise tree shapes on which the indices are identical, and study the equivalence of FP and SV and its implications in more detail. 
We also define and investigate analogues of these indices on unrooted trees (where SV was originally defined), including an index that is 
closely related to the Pauplin representation of phylogenetic diversity.
\end{abstract}

\newpage

\section{Introduction}
\label{intro_sec}

Phylogenetic trees play an important role in  quantifying biodiversity by estimating how much `evolutionary heritage' is captured by each species and thus how much may be lost due to the current high rates of species extinction. The concept that each extant species caries  a combination of unique and shared evolutionary history leads naturally to the notion of a phylogenetic diversity index  for each species, which depends on its placement in the underlying phylogenetic tree, which, when summed together (across all species), gives the total diversity of the tree \citep{Redding2008, Redding2014, Vellend2011}. For example, the reptile species {\em tuatara}, being the sole surviving species from the superorder Lepidosauria,  represents 220 million years of unique evolution as traced back to when this species branched off its phylogenetic tree from other lineages that have survived to the present.  This species also carries further evolutionary history that is shared with other extant species, and phylogenetic diversity indices quantify not only the unique evolutionary history, but shared history as well.
   
Methods to apportion the total evolutionary history of life (measured in time or in  genetic or trait diversity)  across present-day species can be implemented in various ways.  In this paper, we explore the mathematical relationship between three closely related indices.  Two of these indices -- (FP) Fair Proportion \citep{Redding2003} and  (ES) Equal Splits \citep{Redding2006} -- were described for rooted trees, while a third, the Shapley Value (SV), from cooperative game theory, was initially introduced for unrooted trees \citep{Haake2008}.  Soon afterwards it was shown that SV on {\em rooted} trees is actually equivalent to FP \citep{fuc} (see also \cite{sta}).
 These and other related indices, have been incorporated into the EDGE initiative by the  Zoological Society of London \citep{isa07} to quantify
  the expected loss of evolutionary history associated with different endangered species.

The structure of this paper is as follows. We first review some basic definitions, then define two of the indices (FP and ES). Next, we consider how different FP and ES can be from each other. We do this first by considering their ratios (FP/ES and ES/FP) to obtain concise exact results  (Theorem~\ref{thm1}) which apply regardless of whether or not a molecular clock assumption is imposed.   As a simple example of how these results apply, consider all rooted binary phylogenetic trees that classify (say) 20 species at their leaves and all possible assignments of edge lengths.  It is then possible for the ES index of a species to be up to 9 times larger (but no more) than the FP index for that species; on the other hand, the FP index of a species can be up to 13,797 times larger
(but no more) than the ES index of that species.

We then  consider how large the differences FP$-$ES and ES$-$FP can be, where now we need to bound some aspect of the tree length---either the longest edge length (Theorem~\ref{thm2}) or the total length of the tree (Theorem~\ref{thm3}). 
Companion results  are also derived for molecular clock trees. In Theorem~\ref{thm4},
we characterise the set of trees for which FP and ES are identical, and \blue{Section~\ref{sec-shap}} provides a proof that SV is uniquely characterized by four axioms on trees, by using the equivalence of FP and SV.  In Section \ref{sec_un},  we consider variants of FP and ES defined on unrooted trees and establish a number of results for these measures. We end by highlighting some questions for future work.

\subsection{Rooted trees and phylogenetic diversity indices}
\label{roo}
In this section and the next we deal with rooted phylogenetic $X$--trees. A rooted tree $T$ with leaf set $X$ is said to be a (rooted) phylogenetic $X$--tree if each non-leaf vertex  is unlabelled and has out-degree at least 2 (two such trees are considered identical if there is a graph isomorphism between them that sends leaf $x$ to leaf $x$ for each $x \in X$).  In the case where all of the non-leaf vertices have out-degree 2, we say that the tree is {\em binary}; we will  mostly work with this class in these two sections. Background on the basic combinatorics of phylogenetic trees can be found in \cite{Book_Mike}. 
For the rest of this paper we will take, without loss of generality, the leaf set $X$ of trees to be $X=[n]=\{1, \ldots, n\}$, where $n\geq 2$.

Throughout this section, let $T$ be a rooted binary phylogenetic tree with root $\rho$ and leaf set $[n]$, where each edge $e$ is assigned a non-negative length $l(e)$. 
 Let $L=L(T, l) = \sum\limits_el(e)$ be the total sum of edge lengths of $T$ (see Figure \ref{fig1}(a)).

Any function $\varphi_T: [n] \rightarrow \mathbb{R}$ such that $\sum_{i \in [n]} \varphi_T(i) = L(T, l)$ is called a \emph{phylogenetic diversity index} or \emph{PD index} for short. If $\varphi_T(i)$ can be written as a linear function on the edge lengths of $T$, i.e. 
\begin{equation}
\label{lindiv}
\varphi_T(i) = \sum_{e \in E(T)} \gamma_T(i,e) l(e)
\end{equation}
for coefficients $\gamma_T(i,e)$ that are independent of $l(e)$, we call $\varphi_T$ a \emph{linear diversity index}. 
In this paper, we will consider three linear PD indices, namely the Fair Proportion index, the Equal Splits index and the Shapley value.  
Note that an arbitrary function $\varphi_T$ of the form described in Eqn.~(\ref{lindiv}) is a  diversity index if and only if the following linear equations hold for the coefficients 
$\gamma_T(i,e)$, for each edge $e$ of $T$:
\begin{equation}
\label{lindiv2}
\sum_{i \in [n]} \gamma_T(i,e) =1.
\end{equation}

\subsection{Fair Proportion and Equal Splits}
The \emph{Fair Proportion (FP) index} \citep{Redding2003} for leaf $i \in [n]$ (also called `evolutionary distinctiveness') is defined as:
\begin{align}
FP_T(i) = \sum\limits_{e \in P(T; \rho, i)} \frac{1}{n(e)}l(e),
\end{align}
where $P(T; \rho, i)$ denotes the path in $T$ from the root to leaf $i$, $l(e)$ is the length of edge $e$ and $n(e)$ is the number of leaves descended from $e$. Essentially, the FP index distributes each edge length evenly among its descendant leaves. 
Note that as the order of summation in the definition of the FP index does not matter, we will often reverse the order and go from leaf $i$ to the root, since this is common biological practice.
As an example, for the tree $T$ shown in Fig.~\ref{fig1}(a) and the leaf $i=1$, we have $FP_T(i)=\frac{1}{1}+ \frac{1}{2}+\frac{1}{3}= \frac{11}{6}$.

A second natural index is the \emph{Equal Splits (ES) index} \citep{Redding2006}, where each edge length is distributed evenly at each branching point. It is defined as:
\begin{align}
ES_T(i) = \sum\limits_{e \in P(T; \rho, i)} \frac{1}{\Pi(e,i)}l(e),
\end{align}
where $\Pi(e,i) = 1$ if $e$ is a pendant edge incident with $i$; otherwise, if $e=(u,v)$ is an interior edge, then $\Pi(e, i)$ is the product of the out-degrees of the interior vertices on the directed path from $v$ to leaf $i$. 
Since we will be dealing with binary  trees in this paper, $\Pi(e,i)$ is 2 raised to the power of the number of edges between $e$ and leaf $i$. 
As an example, for the tree $T$ shown in Fig.~\ref{fig1}(a), and the leaf $i=1$, we have $ES_T(i)=\frac{1}{1}+ \frac{1}{2}+\frac{1}{4}= \frac{7}{4}$ (where we have again reversed the order of summation).

\begin{figure}[htbp]
\centering
\includegraphics[scale=0.2]{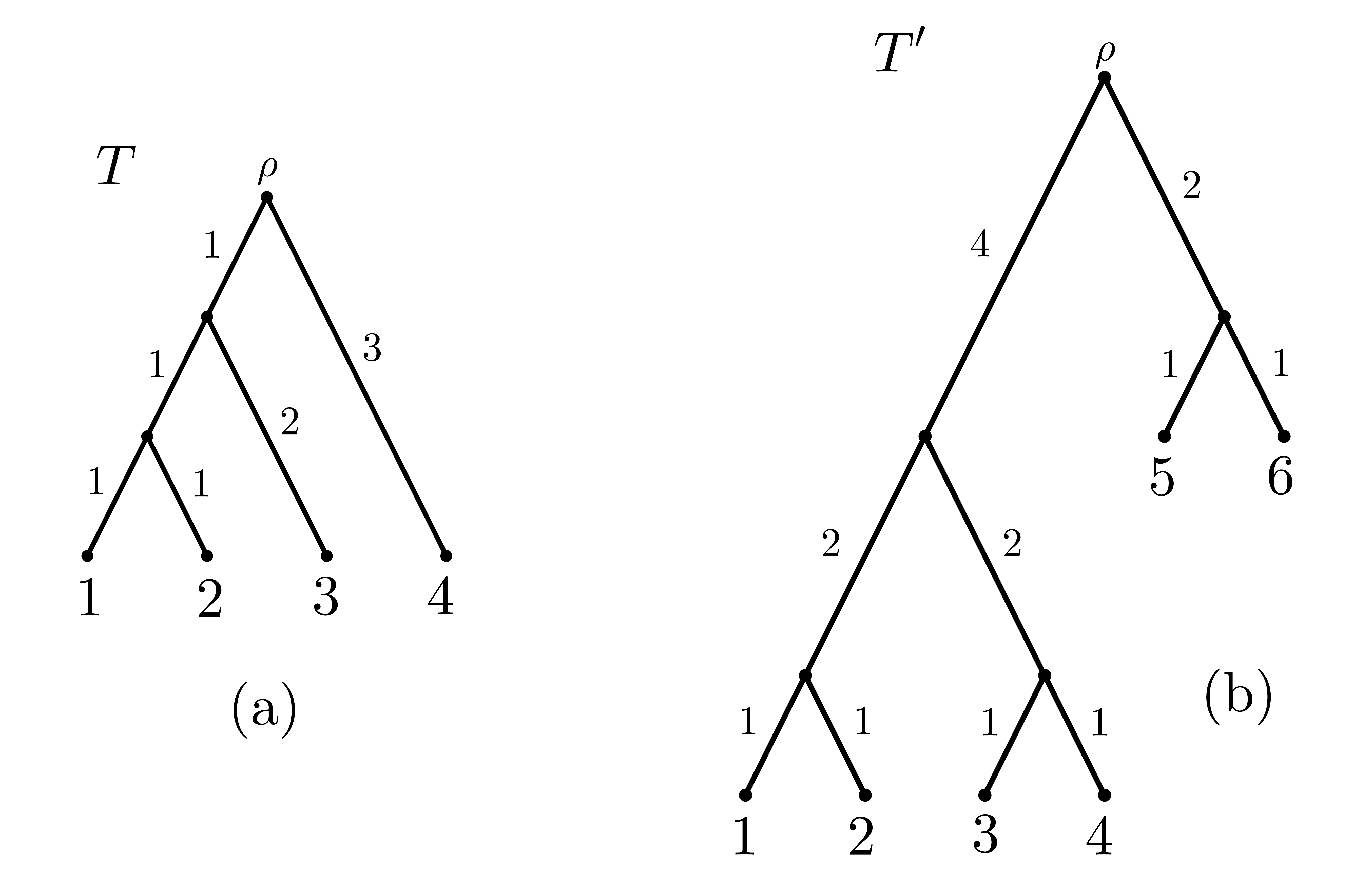}
\caption{(a) A rooted binary phylogenetic tree $T$ on leaf set $[4]$. 
We have $FP_{T}(1)=FP_{T}(2)=11/6, \, FP_{T}(3)=7/3$ and $FP_{T}(4)=3$. Similarly, $ES_{T}(1)=ES_{T}(2)=7/4, \, ES_{T}(3)=5/2$ and $ES_{T}(4)=3$. Note that the edge lengths of $T$ follow a `molecular clock', as all leaves have the same distance to the root. 
(b) A rooted binary phylogenetic tree $T'$ on leaf set $[6]$. For the subset $S=\{1,2,4\} \subseteq [6]$ of leaves, we have $PD_{\blue{T'}}(S) = 1+1+1+2+2+4 = 11$. Moreover, we have $FP_{T'}(i)=3$ for $i=1, \ldots, 4$ and $FP_{T'}(j)=2$ for $j=5,6$. Furthermore, $ES_{T'}(i)=3$ for $i=1, \ldots, 4$ and $ES_{T'}(j)=2$ for $j=5,6$. Thus, $FP_{T'}(i) = ES_{T'}(i)$ for all $i \in [6]$. Note that the edge lengths of $T'$ do not follow a `molecular clock'.}
\label{fig1}
\end{figure}

Both FP and ES are linear diversity indices (in particular, $\sum_{i \in [n]} FP_T(i) = \sum_{i \in [n]} ES_T(i) = L(T,l)$). This is easy to see for FP but is less obvious for ES (it suffices to show that Eqn. (\ref{lindiv2}) holds, which is given by  Lemma~\ref{lemma_sm} later in this paper). 
In general, $FP_T(i) \neq ES_T(i)$, with Figure \ref{fig1}(a) providing a simple example. 
This raises the question of how different FP and ES can be, and under which circumstances they coincide.  Although there have been some simulation studies to compare the two indices on various trees and taxon choices  \citep{Redding2008,Redding2014}, in the first part of this paper, we determine the largest difference possible between one index and the other (both in relative terms and for absolute differences), and also considering the differences when the edge lengths are constrained to be `clock-like' or not.  In particular, rather than considering how different these indices might be `on average' or for a particular tree with particular edge lengths, we study how different they can be for rooted trees in the most extreme cases. 

\section{How different can FP and ES be?}
\label{sec_diff}
In this section, we  investigate the maximal difference (across all binary trees with $n$ leaves and all edge lengths, and all leaf choices) between the Fair Proportion index and the Equal Splits index (and vice versa), both in terms of their ratios and their absolute values.  Before proceeding, we introduce some further notation that will be helpful in the arguments that follow. Let $RB(n)$, $n \geq 2$, denote the set of all binary rooted phylogenetic trees on leaf set $[n]$.  

Notice that each pair $(T, i)$, where $T$ in $RB(n)$,  $i \in [n]$ is a leaf of $T$, gives rise to a uniquely defined directed path $e_h, \ldots, e_0$ from the root $\rho$ of $T$ to leaf $i$.  We will let $n_j$ denote the number of leaves descended from the endpoint of $e_j$ closest to the leaves.
Thus, $n_0=1$  and $n_j \geq j+1$ for all $j>0$. 
In addition, when the edge $e_j$ has an associated non-negative length $l(e_j)$, we will let $l_j$ denote this length.
We will use this notation throughout this paper. 
 In the case where $n_j = j+1$ for all $1\leq j \leq h$ and $n_h = n-1$ (i.e. when each of the pendant subtrees in Fig.~\ref{fig2} has just one leaf), then $T$ is said to be a rooted {\em caterpillar tree}, with $i$ in its {\em cherry} (a cherry is a pair of leaves adjacent to the same vertex). Note that a tree in $RB(n)$ is a caterpillar if and only if it has exactly one cherry.

\begin{figure}[htb]
\centering
\includegraphics[scale=0.8]{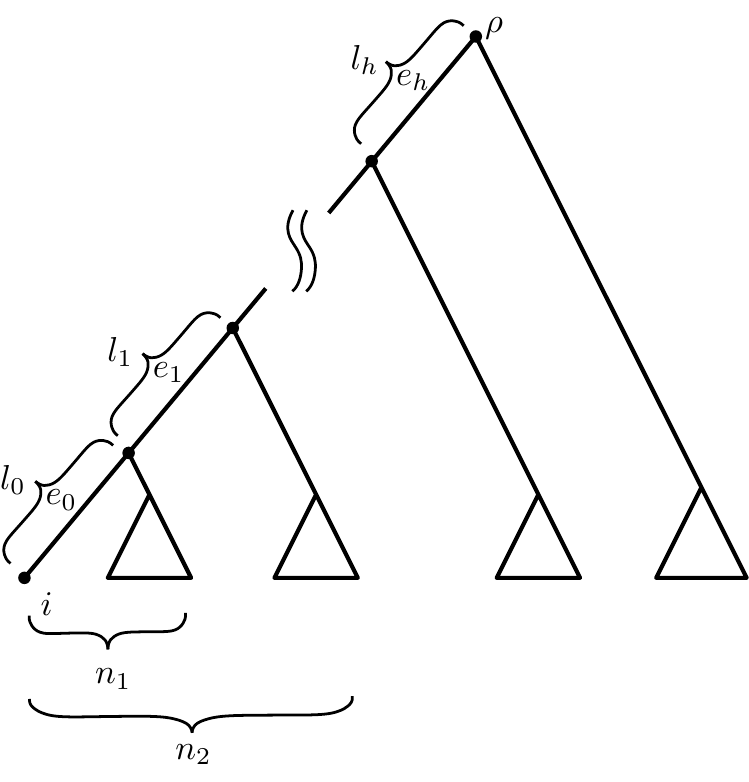}
\caption{Representing a binary tree relative to a reference leaf $i$, where $l_j$ refers to the length of edge $e_j$, and $n_j$ is the number of leaves of $T$ that are descended from the endpoint of $e_j$ that is closest to the leaves. }
\label{fig2}
\end{figure}

We will also occasionally consider a further `molecular clock' condition on the edge lengths: 
	\begin{enumerate}
	\item[] {\rm \bf (MC)} The  sum of the edge lengths from the tree root to leaf $i$ takes the same value for each leaf $i$.
	\end{enumerate}
	This condition applies, for example, if the edge lengths correspond to time, and all the leaves at the tree are sampled at the same time (e.g. at the present; cf. Figure \ref{fig1}(a)).

\subsection{Maximal ratios}

We first consider how large the FP can be relative to ES (i.e. as a ratio), as well as the ratio of ES to FP.  Let
$$\Delta_n(FP/ES) = \max\limits_{T \, \in \, RB(n)} \; \max\limits_{i \, \in \, [n]} \; \sup\limits_{l > 0}  \;  \Big\lbrace \frac{FP_T(i)}{ES_T(i)} \Big\rbrace,$$
and
$$\Delta_n(ES/FP) = \max\limits_{T \, \in \, RB(n)} \; \max\limits_{i \, \in \, [n]} \; \sup\limits_{l > 0}  \;  \Big\lbrace \frac{ES_T(i)}{FP_T(i)} \Big\rbrace,
$$
where (here and below) `sup' refers to supremum (over all assignments $l$ of edge lengths  that are positive).	

In words, $\Delta_n(FP/ES)$ measures the largest possible ratio of the FP index to the ES  index  across all binary trees with $n$ leaves, all choices of leaf $i$, and all assignments of strictly positive edge lengths.
Similarly, $\Delta_n(ES/FP)$ measures the analogous extreme value for the ratio of ES to FP. 
Throughout this paper, we impose strictly positive edge lengths (in taking the supremum), in order to avoid any ambiguity as to whether an edge in a tree with a zero length edge should be contracted (this causes a discontinuity for the ES value), and to avoid any issues associated with  fractions of the form $0/0$.

Our first theorem shows that, in the most extreme case,  the ratio of FP to ES grows exponentially with $n$, whereas the ratio of ES to FP grows only linearly with $n$.

\begin{theorem}
\label{thm1}
For $n\geq 3$:
$$\Delta_n(FP/ES) =  \frac{2^{n-2}}{n-1}  \mbox{ and } \Delta_n(ES/FP) =  \frac{n-1}{2}.$$
		Moreover, these results hold if the molecular clock condition {\bf (MC)} is imposed.
\end{theorem}
	
\begin{proof}

Our proof makes use of the following classical inequality, due to Cauchy (for details, see \cite{stee04}, pp. 82).  
Let $a_i, b_i>0$ be constants for $i=0,1,\ldots, h$.
Then 
\begin{equation}
\label{classicineq}
\frac{\sum_{j=0}^h a_j}{\sum_{j=0}^h b_j} \leq \max \left\{\frac{a_j}{b_j}, j=0, \ldots, h \right\}.
\end{equation}

For the first ratio (FP/ES),  using the notation in Fig.~\ref{fig2}, we have:
$$\frac{FP_{T}(i)}{ES_{T}(i)} = \frac{\sum_{j=0}^h l_j/n_j}{\sum_{j=0}^h l_j/2^j},$$ and since
$n_j \geq j+1$, we have:
\begin{equation}
\label{ratio}
\frac{FP_{T}(i)}{ES_{T}(i)} \leq  \frac{\sum_{j=0}^h l_j/(j+1)}{\sum_{j=0}^h l_j/2^j} \leq \max \left\{\frac{l_j/(j+1)}{l_j/2^j}, j=0, \ldots, h \right\} = \max \left\{\frac{2^j}{j+1}, j=0, \ldots, h \right\},
\end{equation}
where the second inequality is from (\ref{classicineq}).  Now, the expression on the far right of (\ref{ratio}) is maximised  (subject to the constraint  that $j \leq h \leq n-2$) by taking $j=h=n-2$, which gives:
\begin{equation}
\label{f_p}
\frac{FP_{T}(i)}{ES_{T}(i)} \leq \frac{2^{n-2}}{n-1}.
\end{equation}
To see that this bound can be realised (in the supremum limit), consider a caterpillar tree that has leaf $i$ in its cherry and where the edges on the path from $\rho$ to $i$ have strictly positive edge lengths $\ell', \ell, \ldots, \ell$, respectively (see Fig.~\ref{fig3}(a)).  In the limit as the ratio $\ell'/\ell$ tends to infinity, 
$\frac{FP_{T}(i)}{ES_{T}(i)}$ converges to $\frac{2^{n-2}}{n-1}$ which, combined with Inequality~(\ref{f_p}), establishes the first equality in Theorem~\ref{thm1}.
Moreover,  it is clear that one can select the other edge lengths in $T$ so that the {\bf (MC)} condition  holds.

\begin{figure}[htbp]
	\centering
	\includegraphics[scale=0.8]{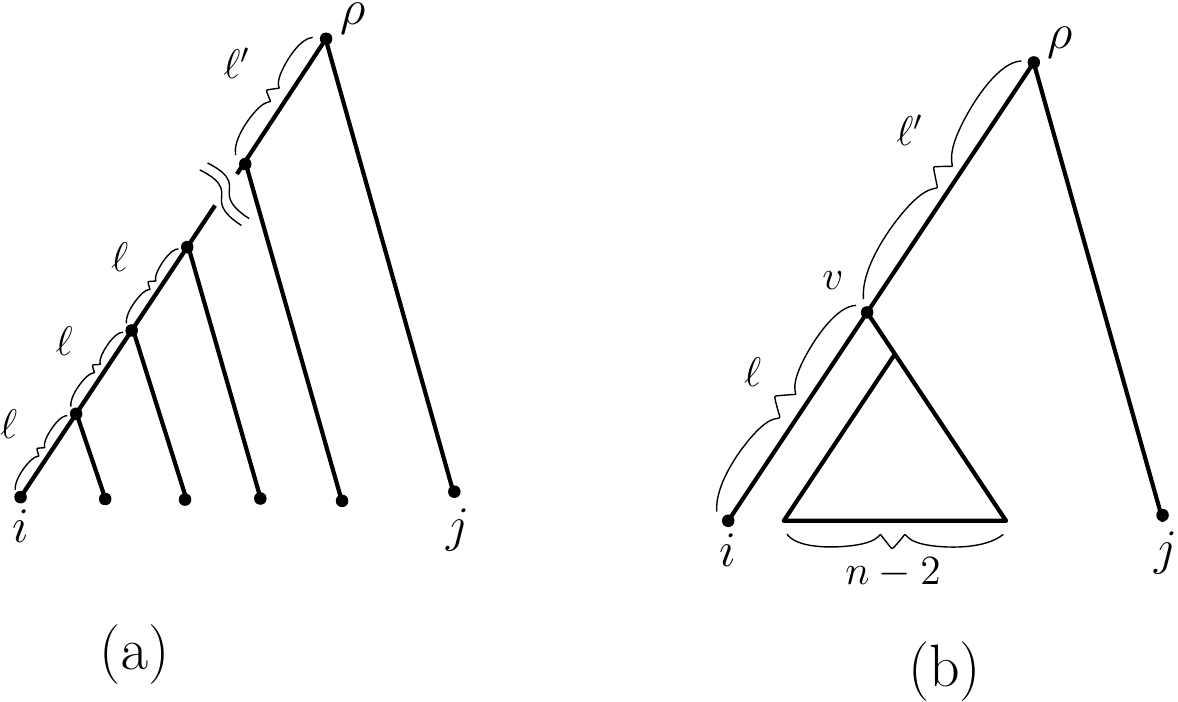}
	\caption{Trees for the proof of Theorem~\ref{thm1}.}
	\label{fig3}
\end{figure}

\bigskip

For the proof of the second equality in Theorem~\ref{thm1}, we have:
\begin{equation}
\label{ratio2}
\frac{ES_{T}(i)}{FP_{T}(i)} = \frac{\sum_{j=0}^h l_j/2^j}{\sum_{j=0}^h l_j/n_j}.
\end{equation}
By Inequality~(\ref{classicineq}), we have:
$$\frac{ES_{T}(i)}{FP_{T}(i)}  \leq \max \left\{\frac{l_j/2^j}{l_j/n_j}, j=0, \ldots, h \right\} = \max \left\{\frac{n_j}{2^j}, j=0, \ldots, h \right\}.$$

Now, $n_0=1$ and for each $j>0$ we have $n_j\leq n-(h-j)-1$.  Subject to these constraints, the ratio $\frac{n_j}{2^j}$ is maximised by
setting $n_1=n-1$ (with $h=j=1$). 
Thus  
\begin{equation}
\label{e_f}
\frac{ES_{T}(i)}{FP_{T}(i)}  \leq \frac{n-1}{2}.
\end{equation}	
To see that this bound can be realised, let $T \in RB(n)$ be such that the children of the root consist of a leaf $j$ and an interior vertex $v$, where the children of $v$ consist of leaf $i$ and a subtree of $T$ having $n-2$ leaves.
Let the edge between the root and $v$ have length $\ell'>0$ and assign length $\ell>0$ to the edge $(v, i)$  (see  Fig.~\ref{fig3}(b)). In the limit as the ratio $\ell'/\ell$ tends to infinity $ \frac{ES_T(i)}{FP_T(i)}$ converges to $\frac{n-1}{2}$ which, combined with Inequality~(\ref{e_f}), establishes the second part of Theorem~\ref{thm1}.
Again, it is clear that one can select the other edge lengths in $T$ so that the {\bf (MC)} condition  holds.

\bigskip

\end{proof} \hfill$\Box$ \\

\subsection{Maximal differences in terms of $l_{\rm max}$.}

In this section and the next, we consider the additive difference between $FP_T(i)$ and $ES_T(i)$  and vice versa for any tree $T \in RB(n)$ and any leaf $i$ of $T$. These differences can be expressed as follows:
\begin{align}
FP_T(i)-ES_T(i) &= \sum\limits_{j=1}^{h} l_j \left(  \frac{1}{n_j} - \frac{1}{2^j} \right) \mbox{ and }  \label{FP_minus_ES} \\
ES_T(i)- FP_T(i) &= \sum\limits_{j=1}^{h} l_j \left( \frac{1}{2^j} - \frac{1}{n_j} \right). \label{ES_minus_FP}
\end{align} 
Note that both sums start at $j=1$, since for $j=0$ we have $n_j=2^j=1$ and so the additional term in either sum that would correspond to $j=0$ is zero.
Also, in  contrast to the ratios considered in the last section, these differences can be arbitrarily large (e.g. multiplying all the edge lengths by a constant $C$ will increase the difference $FP_T(i)-ES_T(i)$ by $C$).  Thus we
 will analyse these maximal differences both in terms of the length of the longest edge of a tree $l_{\rm max} = \max\limits_e l(e)$ and in terms of the sum of edge lengths $L= \sum\limits_e l(e)$.

Our second theorem shows how the absolute  differences between FP and ES (and vice versa) grow either slowly (logarithmically) or are bounded independent of $n$. In particular, the absolute difference between FP-ES can be made arbitrarily large (for a fixed value of $l_{{\rm max}}$) by increasing the number of taxa; however, 
ES$-$FP cannot (it is always bounded above by $l_{{\rm max}}$ regardless of $n$). Moreover, if we impose a molecular clock, then
FP$-$ES now becomes bounded above by a constant times $l_{{\rm max}}$. The situation with absolute differences is thus quite different from that for the ratios FP/ES and ES/FP.

To state the theorem more succinctly, we introduce some additional notation.
Let $$\Delta_n(FP-ES; l_{\rm max}) = \max\limits_{T \, \in \, RB(n)} \; \max\limits_{i \, \in \, [n]} \; \sup\limits_{l > 0: \max\{l(e)\}=l_{{\rm max}}} \; \{FP_T(i) - ES_T(i) \}.$$
In words, $\Delta_n(FP-ES; l_{\rm max})$ is the largest possible difference between FP and ES across the set of 
\begin{itemize}
\item binary trees $T$ with $n$ leaves, and
\item  assignments of positive edge lengths to $T$ that have a maximal edge length $l_{\rm max}$, and 
\item choices of leaf $i$.
\end{itemize}

Similarly, let  $$\Delta_n(ES-FP; l_{\rm max}) = \max\limits_{T \, \in \, RB(n)} \; \max\limits_{i \, \in \, [n]} \;  \sup\limits_{l > 0: \max\{l(e)\}=l_{{\rm max}}} \; \{ES_T(i) - FP_T(i) \}.$$
Note that $\Delta_n(FP-ES; l_{\rm max})=\Delta_n(ES-FP; l_{\rm max}) =0$ for $n =2,3$.   In  the following theorem, we consider the case $n \geq 4$, and we let
$\gamma$ denote the Euler--Mascheroni constant ($\approx 0.5772$), and $o(1)$ denote a term that converges to 0 as $n$ grows.


\begin{theorem}
\label{thm2}
For each $n \geq 4$:
	\begin{enumerate}
		\item[{\rm (i)}]
		\begin{itemize}
		\item[{\rm (a)}] $\Delta_n(FP-ES; l_{\rm max}) =  l_{\rm max} \cdot \left( \ln n  + \gamma- 2 \right) + o(1)$.
		\item[{\rm (b)}]
		$\Delta_n(ES-FP; l_{\rm max})  < l_{\rm max},$ {\rm and }
		$$\sup\limits_{n} \Delta_n(ES-FP; l_{\rm max}) = l_{\rm max}.$$

		\end{itemize}
	\item[{\rm (ii)}]
		If {\rm \bf (MC)} holds, then $\Delta_n(FP-ES; l_{\rm max}) < l_{\rm max} \cdot \frac{2}{\ln 2}$.

	\end{enumerate}
\end{theorem}

\noindent {\em Proof of Part (i--a):} We first show that a triple $(T, i, l)$ that realizes the quantity $\Delta_n(FP-ES; l_{\rm max})$ is a rooted caterpillar tree on $n$ leaves with $i$ being a leaf of the cherry in $T$, and each edge
on the path from the root of $T$ to $i$ having length $l_{\rm max}$. This is illustrated in Fig.~\ref{fig4}(a).
Let $e_j, l_j$ and $n_j$ be as described in Fig.~\ref{fig2}.
Let  $\delta_j = l_j \left( \frac{1}{n_j} - \frac{1}{2^j} \right)$ denote the contribution of edge $e_j$ to $FP_{T}(i)-ES_{T}(i)$ (cf. Eqn. \eqref{FP_minus_ES}). Using only the fact that  $T \in RB(n)$  it follows that $n_j \geq j+1$ for each $j \geq 0$  and so $\delta_j \leq l_j \left( \frac{1}{j+1} - \frac{1}{2^j} \right)$. In particular, $\argmax\limits_{n_j} \; \left\lbrace l_j \left( \frac{1}{n_j} - \frac{1}{2^j} \right) \right\rbrace \, = \left\lbrace \frac{1}{j+1} \right\rbrace$, and so  $\delta_j$ is maximal if and only if $n_j = j+1$.
As this holds for all values of $j$, this immediately implies that the maximal pending subtree of $T$ containing leaf $i$ (call it $t_1$) has to be a caterpillar tree on $n' \leq n-1$ leaves and with $i$ being a leaf of the cherry of this caterpillar.
We show that $n'=n-1$ (and thus $T$ is a caterpillar) by deriving a contradiction. Suppose that $n' < n-1$.  In that case, the two subtrees of $T$ incident with the root of $T$ consist of $t_1$ and another subtree (call it $t_2$) that has two or more leaves.
 In particular, this implies that $h < n-2$ (i.e. there are less than $n-1$ edges on the path from $i$ to the root of $T$). However, as $\frac{1}{j+1} - \frac{1}{2^j} \geq 0$ for each $j$, this would imply that $T$ is not a tree that maximises  $\max\limits_{i' \, \in \, [n]} \; \sup\limits_{l > 0} \; \{FP_T(i') - ES_T(i') \}$, since $FP_{T}(i)-ES_{T}(i)$ could be increased by sequentially attaching all but one leaf from $t_2$ to the edge connecting $t_1$ and the root (i.e. by extending the length of the path from leaf $i$ to the root of $T$). Thus, $n'=n-1$, and therefore $T$ has to be the caterpillar tree on $n$ leaves  that has $i$ in its cherry.
Moreover,  by again invoking the inequality $\frac{1}{j+1} - \frac{1}{2^j} \geq 0$ (for all $j \geq 0$) and recalling that $\delta_j = l_j \left (\frac{1}{j+1} - \frac{1}{2^j} \right)$, we can also conclude that $l_j = l_{\rm max}$ for all $j$ (as otherwise $\delta_j$ and thus, $FP_{T}(i)-ES_{T}(i)$ could be increased). In summary, $(T, i, l)$ has the structure claimed.
\begin{figure}[htb]
	\centering
	\includegraphics[scale=0.5]{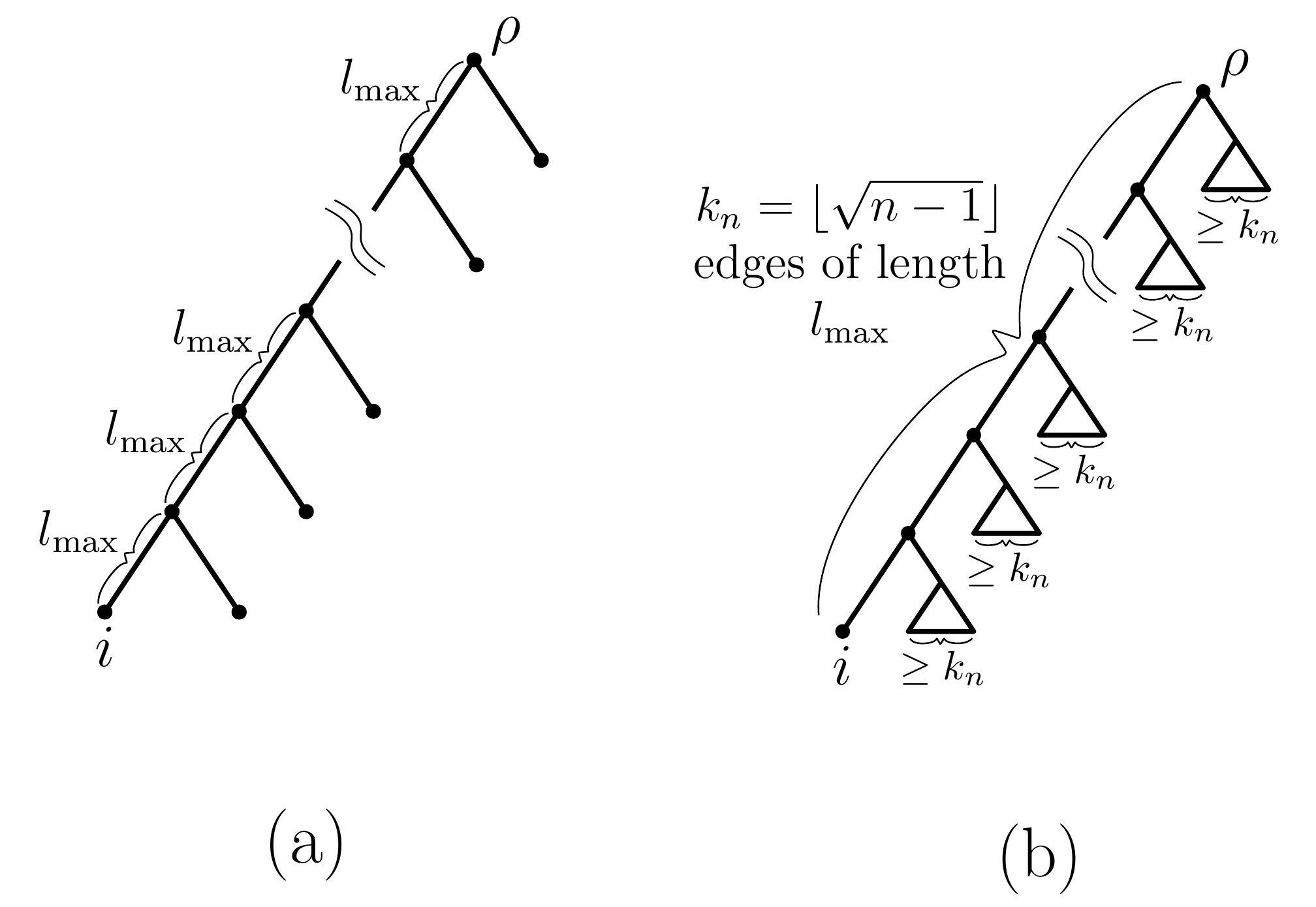}
	\caption{Trees for the proof of Theorem~\ref{thm2}.}
	\label{fig4}
\end{figure}
It is now straightforward to calculate  $\Delta_n(FP-ES; l_{\rm max})$  for the optimal choice of $(T, i, l)$ described above. We have:
\begin{align*}
FP_T(i) &= l_{\rm max} \cdot \sum\limits_{j=0}^{n-2} \frac{1}{j+1} =  l_{\rm max} \cdot  \left( \ln n + \gamma + o(1) \right),\\
{\rm and}\\ 
ES_T(i) &= l_{\rm max} \cdot  \sum\limits_{j=0}^{n-2} \frac{1}{2^j} = 2 \, l_{\rm max} + o(1).
\end{align*}
Consequently,
$$\Delta_n(FP-ES;  l_{\rm max}) =l_{\rm max} \cdot \left( \ln n + \gamma - 2 \right) + o(1),$$
which completes the proof of Part (i--a).

\bigskip

{\em Proof of Part (i--b):} From Eqn. \eqref{ES_minus_FP}, we have:
\begin{equation}
\label{est}
ES_T(i) - FP_T(i) = \sum\limits_{j=1}^{h} l_j \left( \frac{1}{2^j} - \frac{1}{n_j} \right) < \, \sum\limits_{j=1}^{h} l_j \frac{1}{2^j} \, \leq \, l_{\rm max} \cdot \sum\limits_{j=1}^{h} \frac{1}{2^j} \, < l_{\rm max}.
\end{equation}
Thus, $\Delta_n(ES-FP; l_{\rm max})  < l_{\rm max}.$  To show that
$\sup_n \Delta_n(ES-FP; l_{\rm max})  = l_{\rm max},$  let $T_n$ be a tree in which the path $P$ from the root to leaf $i$ has
$k_n = \lfloor \sqrt{n-1}\rfloor$ edges, and each of the subtrees incident with the vertices of $P$ (except the final leaf vertex) 
has at least $k_n$ leaves.  Assign edge length $l_{\rm max}$ to each of the edges in $P$.   This is illustrated in Fig.~\ref{fig4}(b). 
Then 
$$ES_{T_n}(i) - FP_{T_n}(i) =l_{\rm max} \cdot \sum\limits_{j=1}^{k_n}  \left( \frac{1}{2^j} - \frac{1}{n_j} \right).$$
Now, $\lim_{n \rightarrow \infty} \sum\limits_{j=1}^{k_n}  \frac{1}{2^j}  = 1$
and since $n_j \geq j \cdot k_n$, we have: 
$$\sum\limits_{j=1}^{k_n}   \frac{1}{n_j}  \leq \frac{1}{k_n}\cdot \sum\limits_{j=1}^{k_n}  \frac{1}{j} \sim 
\frac{\ln(k_n)}{k_n} \rightarrow 0,$$
as $n \rightarrow \infty$.
Combining this with Eqn.~(\ref{est}) gives:
$\lim_{n\rightarrow \infty} ES_{T_n}(i) - FP_{T_n}(i) = l_{\rm max}$, as required.

\bigskip

{\em Proof of Part (ii):}  Let $T \in RB(n)$ and $i\in [n]$. From Eqn. \eqref{FP_minus_ES}, we have $FP_T(i)-ES_T(i) = \sum\limits_{j=1}^h l_j \left( \frac{1}{n_j} - \frac{1}{2^j} \right).$
We claim that, under condition {\rm \bf (MC)}, 
\begin{equation}
\label{d1}
n_j \geq 2^{\lceil \sum_{k=0}^{j-1}l_k / l_{\rm max} \rceil}.
\end{equation}
To establish Inequality (\ref{d1}), the {\rm \bf (MC)} condition implies that for each leaf $i'$ of $T$ descended from the endpoint $v_j$ of $e_j$ closest to the leaves, the sum of the edge lengths from $v_j$ to leaf $i'$ is
equal to $\sum_{k=0}^{j-1}l_k$.  Moreover, each of these edges has length at most $l_{\rm max}$, which means that the number of edges on this path must be at least $m \coloneqq \lceil \sum_{k=0}^{j-1}l_k / l_{\rm max} \rceil$. 
Now, for $r \geq 1$, let $N_r$ be the number of vertices descended from $v_j$ that are separated from $v_j$ by exactly $r$ edges. We then have $N_r = 2^r$ for all $r=1, \ldots, m$. This follows from an inductive argument. Clearly, $N_1=2$ (as $T$ is binary and $v_j$ is not a leaf since $j > 0$). Suppose the statement is true for $1 \leq r < m$ and consider $N_{r+1}$. Each vertex counted by $N_r$ must have two children (otherwise there would be a leaf that is separated from $v_j$ by less than $m$ edges) and thus $N_{r+1} = 2 N_r  = 2 \cdot 2^r = 2^{r+1}$, which completes the inductive step. Now, as all leaves descended from $v_j$ are separated by at least $m$ edges from $v_j$, we have $n_j \geq N_m \geq 2^m = 2^{\lceil \sum_{k=0}^{j-1}l_k / l_{\rm max} \rceil}$, which completes the proof.\\

Thus, from Eqn.~(\ref{d1}) and Eqn.~(\ref{FP_minus_ES}), we have:
\begin{equation}
\label{d3}
FP_T(i) - ES_T(i)  \leq \sum_{j=1}^h l_j (2^{- \lceil \sum_{k=0}^{j-1}l_k / l_{\rm max} \rceil} -2^{-j}).
\end{equation}

To complete the proof  of Part (ii), we require the following lemma, the proof of which is provided in the Appendix.

\begin{lemma} \label{Lemma_Convergence}
Suppose that $x_0, x_1, x_2, \ldots, x_h$ all lie in the interval $[0,1]$. Then
$$\sum_{i=1}^h x_i 2^{-\sum_{j<i} x_j} \leq  \frac{2}{\ln 2}\cdot 2^{-x_0}.$$
\end{lemma}

We apply this lemma by setting $x_i = l_i/l_{\rm max}$ for $i =0,1, \ldots, h$.
By Inequality (\ref{d3}), we have:
$$FP_T(i) - ES_T(i)  <   l_{\rm max} \sum_{j=1}^h x_j 2^{- \sum_{k=0}^{j-1}x_k} \leq l_{\rm max} \cdot \frac{2}{\ln 2},$$
as required, where the last inequality is from  Lemma~\ref{Lemma_Convergence}.
 \hfill$\Box$ \\

\subsection{Maximal differences in terms of $L$}

We now describe the maximal possible (positive and negative) difference between FP and ES in terms of the total length of the tree ($L=\sum\limits_e l(e)$), rather than in terms of $l_{\rm max}$ (this is summarized in Theorem~\ref{thm3} below).  Let 
$$\Delta_n(FP-ES; L) = \max\limits_{T \, \in \, RB(n)} \; \max\limits_{i \, \in \, [n]}  \sup\limits_{l > 0: \sum\limits_e l(e)=L}  \; \{FP_T(i) - ES_T(i) \}.$$
In words, $\Delta_n(FP-ES; L)$ is the largest possible difference between FP and ES across the set of:
\begin{itemize}
\item binary trees $T$ with $n$ leaves, and
\item assignments of positive edge lengths to $T$ for which the total sum of the edge lengths is $L$, and
\item choices of leaf $i$.
\end{itemize}
Similarly, let 
$$\Delta_n(ES-FP; L) = \max\limits_{T \, \in \, RB(n)} \; \max\limits_{i \, \in \, [n]}  \sup\limits_{l > 0: \sum\limits_e l(e)=L}  \; \{ES_T(i) - FP_T(i) \}.$$
\begin{theorem}
\label{thm3}
\mbox{}
\begin{itemize}
\item[{\rm (i)}]
$$\Delta_n(FP-ES; L)  = \lambda_n L,$$
where $$\lambda_n= \begin{cases}
0, & \mbox{ for $n=2, 3$}\\
\frac{1}{12}, & \mbox{ for $n=4$} \\
\frac{1}{8}, & \mbox{ for $n=5$}\\
\frac{11}{80}, & \mbox{ for $n \geq 6$},
\end{cases}
$$
and for $n \geq 3$:
$$\Delta_n(ES-FP; L)  = \left(\frac{1}{2}- \frac{1}{n-1}\right)L.$$
\item[{\rm (ii)}]
If the molecular  clock {\bf (MC)} condition is imposed then the above expressions for $\Delta_n(FP-ES; L)$
and $\Delta_n(ES-FP; L)$  remain true if $L$ is replaced by $L/2$.
\end{itemize}
\end{theorem}

\begin{proof}
\mbox{}
For Part (i),  we first  show that for any given tree $T 
\in RB(n)$ and any leaf $i$ of $T$ we have:
\begin{equation}
\label{uv1}
\sup_{l>0, \sum_{e}l(e) = L}  \{FP_{T}(i)-ES_T(i)\} \leq \lambda_n L.
\end{equation}
Recall from Eqn. \eqref{FP_minus_ES} that
$FP_{T}(i)-ES_T(i)= \sum_{j=1}^h l_j \left(\frac{1}{n_j} -\frac{1}{2^j}\right)$,
and observe that  for $j=0$, we have $n_j=1=2^j$. Thus, in particular, we have:
$$FP_{T}(i)-ES_T(i) \leq \max_{0 \leq j \leq h}\left\{\left(\frac{1}{n_j} 
-\frac{1}{2^j}\right)\right\} \cdot \sum_{j=0}^h l_j  \leq \max_{0 \leq j \leq h}\left\{\left(\frac{1}{n_j} -\frac{1}{2^j}\right)\right\} \cdot L.$$
Moreover, for any tree $T$, we always have: $n_j \geq j+1$ for each 
$j\geq 0$, and therefore $$FP_{T}(i)-ES_T(i) \leq  \max_{j\geq 
0}\left\{\left(\frac{1}{j+1} -\frac{1}{2^j}\right)\right\} \cdot L.$$
Let $c_j:=\frac{1}{j+1}-\frac{1}{2^j}$ for $j \geq 0$. The sequence 
$c_j$ for $j=0,1,2, \ldots$ begins as follows:
$$0, 0, \frac{1}{12}, \frac{1}{8}, \frac{11}{80} =0.1375,$$
after which the values in the sequence begin to decline.  This 
establishes Inequality~(\ref{uv1}), as required.

To show that Inequality (\ref{uv1}) is an equality, it suffices to show that for each $n\geq 2$ and every $\epsilon>0$ there 
exists a tree $T \in RB(n)$ with positive edge lengths and there is a leaf $i$ 
of $T$ for which
$FP_{T}(i)-ES_T(i) \geq \lambda_nL-\epsilon$.
To this end, let $T_n$ be a rooted caterpillar tree with $n$ leaves, let $i$ be a leaf 
in the cherry of $T_n$, let the interior edge at distance
$k=\min\{4, n-2\}$ from leaf $i$ have length $L-\epsilon$, and the lengths of all 
the remaining edges of $T_n$ have strictly positive lengths that sum to 
$\epsilon$. In this case:
$$FP_{T}(i)-ES_T(i) \geq \left(\frac{1}{k+1} - 
\frac{1}{2^{k}}\right)\cdot (L-\epsilon)  \geq \lambda_n L 
-\epsilon,$$
holds for $T=T_n$ as required.

\bigskip

We turn now to $\Delta_n(ES-FP; L)$.
We first show that for any given tree 
$T \in RB(n)$ and any leaf $i$ of $T$:
\begin{equation}
\label{eee}
\sup_{l>0, \sum_{e}l(e) = L}  \{ES_{T}(i)-FP_T(i)\} \leq 
\left(\frac{1}{2}- \frac{1}{n-1}\right)L.
\end{equation}
 From Eqn.~\eqref{ES_minus_FP}, we have:
$ES_{T}(i)-FP_T(i)= \sum_{j=1}^h l_j \left(\frac{1}{2^j}- 
\frac{1}{n_j}\right).$
Now, $\left(\frac{1}{2^j}- \frac{1}{n_j}\right)$ takes a value that is, at most,
$\frac{1}{2} - \frac{1}{n-1}$ for all $j \geq 1$. Thus:
$$ES_{T}(i)-FP_T(i) \leq \left(\frac{1}{2} - \frac{1}{n-1}\right) 
\sum_{j=1}^h l_j \leq \left(\frac{1}{2} - \frac{1}{n-1}\right)L,$$
as required to establish Inequality~(\ref{eee}).

To show that Inequality~(\ref{eee}) is an equality it  suffices to show that for each $n \geq 3$, and every $\epsilon>0$ there 
exists a tree $T \in RB(n)$ with positive edge lengths, and there is a leaf $i$ 
of $T$ for which
$$ES_{T}(i)-FP_T(i) \geq \left(\frac{1}{2}- \frac{1}{n-1}\right) L  -\epsilon.$$
To this end, let $T \in RB(n)$ be any tree for which the children of 
the root consist of a leaf $j$ and an interior vertex $v$, where the 
children of $v$ consist of a leaf $i$ and a subtree of $T$ having $n-2$ 
leaves. Let the edge between the root and $v$ have length $L-\epsilon$
and let the remaining edges have strictly positive lengths that sum to 
$\epsilon$.  Then
$$ES_{T}(i)-FP_T(i) = \left(\frac{1}{2}- \frac{1}{n-1}\right) 
(L-\epsilon) \geq  \left(\frac{1}{2}- \frac{1}{n-1}\right) L 
-\epsilon,$$
as required.

\bigskip

{\em Part (ii):} 
We now impose the {\bf (MC)} condition. For $\Delta_n(FP-ES; L)$, observe that our proof of Inequality \eqref{uv1}  invoked the inequality $\sum_{j=0}^h l_j \leq L$. When {\bf (MC)} holds,
we have a tighter bound of the sum, namely $\sum_{j=0}^h l_j \leq L/2$ since there is at least one other leaf $k$ of $T$ for which the path from the root of $T$ to $k$ also has length $\sum_{j=0}^h l_j$ (by {\bf (MC)}) and is edge-disjoint from the path from $\rho$ to $i$ (thus $2\sum_{j=0}^h l_j \leq L$).  
In this way, we claim that  $$\Delta_n(FP-ES; L)  \leq \lambda_nL/2,$$ when {\bf (MC)} holds.

To show that  this  inequality  holds 
 it suffices to show that for each $n \geq 2$, and every $\epsilon>0$ there 
exists a tree $T_n \in RB(n)$ with positive edge lengths, and there is a leaf $i$ 
of $T$ for which:
\begin{equation}
\label{geqineq}
FP_{T}(i)-ES_T(i) \geq \lambda_n L/2-O(\epsilon),
\end{equation}
where $O(\epsilon)$ is a term that tends to zero as $\epsilon \rightarrow 0$.
This trivially holds for $n=2,3$ (indeed it holds for $\epsilon=0$); while for $n=4,5,6$ let $T_n$ be a caterpillar tree with $i$ being a leaf in its cherry.
Let $(\rho, v)$ and $(\rho, j)$ denote the two edges incident with the root $\rho$ of $T$, where $j$ is a leaf. 
Assign the edge $(\rho, v)$  length $L/2 - 5\epsilon / 2$ and the edge $(\rho, j)$ length $L/2 - \epsilon$. We then assign the path from $v$ to $i$
and from $v$ to its adjacent leaf (which exists since it is a caterpillar) a length of $3 \epsilon / 2$. Now adjust the remaining edge lengths so they sum to $\epsilon/2$ and so that the {\bf (MC)} condition holds for $T$ (see Fig.~\ref{fig5}(a) for the case $n=6$). This assignment then satisfies Inequality~(\ref{geqineq}) for $T=T_n$, as required.

\begin{figure}[htbp]
	\centering
	\includegraphics[scale=0.8]{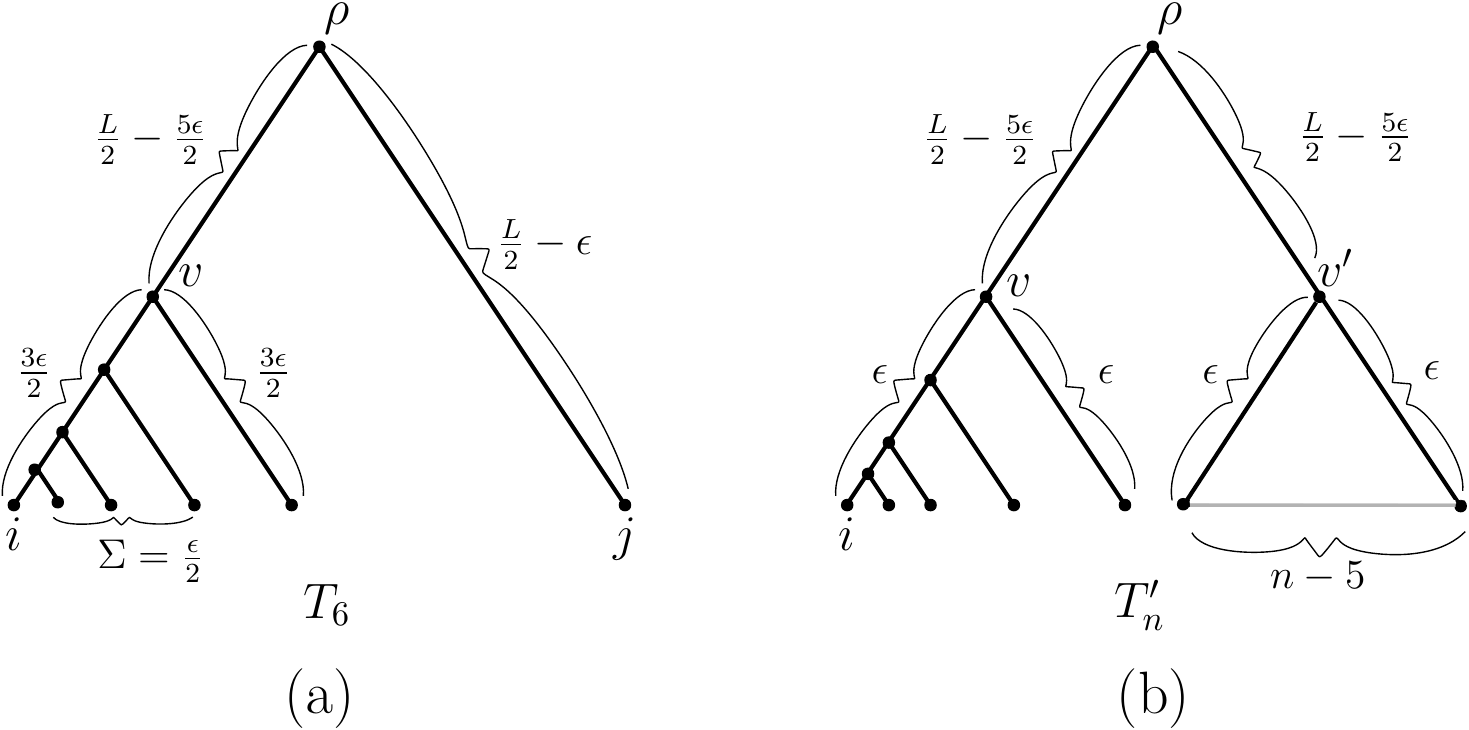}
	\caption{Trees for the proof of Theorem~\ref{thm3}.}
	\label{fig5}
\end{figure}

For the  case $n>6$, let $T'_n$ be obtained from $T_6$ (in the previous argument) by replacing leaf $j$ by an arbitrary rooted binary subtree with $n-5$ leaves with root $v'$. Assign length $L/2-5\epsilon/2$ to each of the two edges 
($(\rho, v)$  and $(\rho, v')$) that are incident with the root.  Set the length of the  path from $v$ to leaf $i$, and the length of the path from $v$ to its adjacent leaf to  equal $\epsilon$, and set the length of each of two disjoint paths from $v'$ to some pair of descendant leaves  also equal to $\epsilon$  (see Fig.~\ref{fig5}(b)).  Finally, 
 select edge lengths within these two subtrees so as to maintain the {\bf (MC)} condition and so that the sum of the lengths of the additional edges added to  these two subtrees is $\epsilon$. In this way, the {\bf (MC)} condition holds for the tree, $L  = 2(L/2-5\epsilon/2)+ 2(2\epsilon) + \epsilon$ equals the sum of the edge lengths, and  Inequality~(\ref{geqineq}) holds for $T=T_n'$, as required.

We now establish Part(ii) for the quantity $\Delta_n(ES-FP; L)$. The argument for the inequality $\Delta_n(ES-FP; L)  \leq \left(\frac{1}{2}- \frac{1}{n-1}\right)L/2$ when {\bf (MC)} holds is identical to the corresponding inequality  for $\Delta_n(FP-ES; L)$ under {\bf (MC)}. 
Moreover, to show that this inequality can be realised, consider again the tree $T \in RB(n)$  described in the previous paragraph, to which we will assign similar but modified edge lengths (we can assume that $n \geq 4$, since the equality holds when $n=3$).  For the edge between the root and $v$, assign length $L/2-\epsilon$; for the edge between the root and leaf $j$,  assign length $L/2-2\epsilon/5$;  for the edge $(v,i)$ assign length $3\epsilon/5$ and assign the lengths of the remaining edges  so that they sum to $4\epsilon/5$ and are chosen so as to satisfy {\bf (MC)} (this is possible, since we are assuming that $n\geq 4$).  In this way, the total sum of edge lengths is $L$ and the path length from the root to each leaf takes the same value (namely,  $L/2-2\epsilon/5$), and the result of Part (ii) for $\Delta_n(ES-FP; L)$ now follows. 
\end{proof} \hfill$\Box$ \\


\section{For which tree shapes do FP and ES coincide?}
\label{sec3}
In the following, we will analyse for which tree shapes FP and ES coincide. Therefore, recall that a rooted binary tree $T$ can be decomposed into its two maximal pending subtrees $T'$ and $T''$ rooted at the direct descendants of the root. We denote this by writing $T=(T', T'')$ (note that the order of $T'$ and $T''$ is not important, thus $T=(T',T'')= (T'', T')$).
Now, let $T$ be a binary tree with  $n=2^h$ leaves, in which each leaf is separated from the root by a path of precisely $h$ edges. We call this (unique shape) tree the \emph{fully balanced tree of height $h$} and denote it by $T_h^{fb}$. Note that we have $T_h^{fb} = (T_{h-1}^{fb}, T_{h-1}^{fb})$, i.e. both maximal pending subtrees of a fully balanced tree of height $h$ are fully balanced trees of height $h-1$. Using the notation of Fig. \ref{fig2} it is thus easy to see that for a leaf $i$ of $T^{fb}_h$ and an edge $e_j$ on the path from the root of $T_h^{fb}$ to leaf $i$ we always have: $n_j=2^j$.  It is now not difficult to show that FP and ES coincide (for all choices of reference leaf $i$) on any fully balanced tree. However, there are other tree shapes for which FP and ES coincide (e.g. the tree $T'$ in Fig.~\ref{fig1}(b)). Therefore, let $T^{sb}$ be a rooted binary tree, whose two maximal pending subtrees $T'$ and $T''$ are both fully balanced trees of height $h'$ and $h''$, respectively (where $h'$ and $h''$ are not necessarily identical), i.e. $T^{sb}=(T^{fb}_{h'}, T^{fb}_{h''})$. We call such a tree a \emph{semi-balanced tree}. Then,

\begin{theorem}
\label{thm4}
Let $T$ be a rooted binary phylogenetic tree on taxon set $[n]$ and
non-negative edge lengths $l(e)$. Then, we have:
$ES_T(i) = FP_T(i)$ for all $i \in [n]$ and all assignments of positive edge
lengths if and only if $T$ is a semi-balanced tree.
\end{theorem}

\begin{proof}
We first show that  if $T$ is a semi-balanced tree (i.e. $T=T^{sb}$) we have
$ES_T(i) = FP_T(i)$ for all $i \in [n]$. Therefore, let $T^{fb}_{h'}$ and
$T^{fb}_{h''}$ denote the two maximal pending subtrees of $T$.
Recall that
$$ FP_T(i) = \sum_{e \in P(T; \rho, i)} \frac{1}{n(e)}l(e) \mbox{ and }
ES_T(i) = \sum_{e \in P(T; \rho, i)} \frac{1}{\Pi(e,i)} l(e). $$
As both sums just run over edges on the path from the root to leaf $i$,
$FP_T(i)$ and $ES_T(i)$ are independent of $T^{fb}_{h''}$ if $i \in
T^{fb}_{h'}$ and vice versa.
Let $i$ be a leaf of
$T^{fb}_{h'}$.
As $T^{fb}_{h'}$ is a fully balanced tree, we have $n_j=2^j$ for all $j=1,
\ldots, h'$, and thus, using Eqn. \eqref{FP_minus_ES}, we immediately have
$$FP_T(i)-ES_T(i) = \sum\limits_{j=1}^{h'} l_j \left( \frac{1}{n_j} -
\frac{1}{2^j} \right) = 0,$$
(i.e. $FP_T(i) = ES_T(i)$).
Analogously, this holds for all leaves of $T^{fb}_{h''}$,  so $ES_T(i) =
FP_T(i)$ for all $i \in [n]$. \\

\bigskip

Now suppose that $FP_T(i)=ES_T(i)$ for all $i \in [n]$.
By way of contradiction assume that $T=(T',T'')$ is not a semi-balanced
tree, i.e. assume that at least one of the maximal pending subtrees of $T$,
say $T'$, is not a fully balanced tree.
This implies that there exists an interior vertex $v$ in $T'$ with the
following two properties:
\begin{enumerate}
   \item[{\rm (i)}] For the subtree $T_v = (T'_v, T''_v)$ rooted at $v$ we
have: $n'_v \neq n''_{v}$, where $n'_v$ and $n''_{v}$ denote the number of
leaves of $T'_v$ and $T''_v$, respectively.
   \item[{\rm (ii)}] $v$ is chosen so that $T_v$ is a minimal subtree of
$T'$ satisfying property {\rm (i)} (in the sense that there exists no
subtree $T_{w}$ of $T'$ on fewer leaves that has this property).
\end{enumerate}
In particular, this implies that both maximal pending subtrees of $T_v$ are
fully balanced trees.
Without loss of generality we may assume that $n'_v > n''_v$ (otherwise
exchange the roles of $T'_v$ and $T''_v$ ), in which case $h'_v > h''_v$.

Now, for a leaf $i$ and an edge $e$ of $T$, we use $\delta_{e}^{FP}(i)$ and
$\delta_{e}^{ES}(i)$ to denote the contribution of edge $e$ to $FP_T(i)$,
respectively $ES_T(i)$, where
$$ \delta_{e}^{FP}(i) = \begin{cases}
   \frac{l(e)}{n(e)}, \, \text{ if } \, e \in P(T; \rho, i); \\
   0, \text{ otherwise}.
   \end{cases} \text{ and }
\delta_e^{ES}(i) = \begin{cases}
   \frac{l(e)}{\Pi(e,i)} \, \text{ if } \, e \in P(T; \rho, i); \\
   0, \text{ otherwise}.
   \end{cases}$$

Let $\Delta^{FP}(i) = \sum_{e \in T_v} \delta_e^{FP}(i)$ and $\Delta^{ES}(i)
=  \sum_{e \in T_v} \delta_e^{ES}(i)$.
Now, as both maximal pending subtrees of $T_v$ are fully balanced trees, we
can use the first part of the proof to conclude that for each $i \in T_v$:
$\Delta^{FP}(i) = \Delta^{ES}(i)$ and we denote this common value by
$\Delta(i)$.

Now, let $l_1, l_2, \ldots, l_h$ be the lengths of the edges $e_1, e_2,
\ldots, e_h$ on the path from vertex $v$ to the root and let $n_j$ be the
number of leaves descended from edge $e_j$.   Let $i'$ be a leaf of $T'_v$
and let $i''$ be a leaf of $T''_v$. By assumption, $FP_T(i)=ES_T(i)$ for all
$i \in [n]$, and so we have
   \begin{align*}
   FP_T(i') &= \Delta(i') + \sum\limits_{j=1}^h \frac{l_j}{n_j} =
\Delta(i') +  \sum\limits_{j=1}^h \frac{l_j}{2^{h'_v + j}} = ES_T(i'),
\text{ and } \\
   FP_T(i'') &= \Delta(i'') + \sum\limits_{j=1}^h \frac{l_j}{n_j} =
\Delta(i'') + \sum\limits_{j=1}^h \frac{l_j}{2^{h''_v + j}} = ES_T(i'').
   \end{align*}
In particular
\begin{align*}
\sum\limits_{j=1}^h \frac{l_j}{n_j} &= \sum\limits_{j=1}^h
\frac{l_j}{2^{h'_v + j}} \mbox{ and } \sum\limits_{j=1}^h \frac{l_j}{n_j} =
\sum\limits_{j=1}^h \frac{l_j}{2^{h''_v + j}}.
\end{align*}

However, as $h'_v > h''_v$ and $l_j > 0$ for all $j$, this is a
contradiction. 
A similar argument yields a contradiction for the assumption that $T''$ is
not a fully balanced tree. Thus, $T$ has to be a semi-balanced tree, which
completes the proof.
\end{proof} \hfill$\Box$ \\ 

\section{Uniqueness of SV for phylogenetic tree games}
\label{sec-shap}
Another linear PD index frequently used is the so-called Shapley value (SV), which originates from cooperative game theory.
Recall that a cooperative game is a pair $([n], \nu)$ consisting of a set of players $[n]=\{1, \ldots, n\}$ and a characteristic function $\nu: 2^{[n]} \rightarrow \mathbb{R}$ that assigns a real value to all subsets of $[n]$ with $\nu(\emptyset) = 0$. A function $\varphi_{\nu} : [n] \rightarrow \mathbb{R}$ that assigns a payoff to each player is called a \emph{value} for the game. 
One such value is the Shapley value (\citet{Shapley1953}), which is defined as follows:
\begin{equation}
\varphi_{\nu}(i) = \frac{1}{n!} \sum\limits_{S \subseteq [n]: \, i \in S} (|S|-1)! \left(n-|S|)! \left( \nu(S) - \nu(S \setminus \{i\} \right)\right). 
\label{def_sv}
\end{equation}

Note that the Shapley value of a player $i$ reflects the average marginal contribution of $i$ to the game. Moreover, it is characterised by the following four axioms: 
\begin{enumerate}
\item Pareto efficiency: $\sum_{i \in [n]} \varphi_{\nu}(i) = \nu([n])$.
\item Symmetry: $\forall \, i, j$ with $i \neq j$ and $\forall \, C \subseteq [n] \setminus\{i,j\}$, if $\nu(C \cup \{i\}) = \nu(C \cup \{j\})$, then $\varphi_{\nu}(i) = \varphi_{\nu}(j)$.
\item Dummy axiom: If $ \forall \, C \subseteq [n] \setminus \{i\}$, $\nu(C \cup \{i\}) = \nu(C)$, then $\varphi_{\nu}(i)=0$.
\item Additivity: $\forall \, \nu_1, \nu_2, \forall i \in [n], \, \varphi_{\nu_1 + \nu_2}(i) = \varphi_{\nu_1}(i) + \varphi_{\nu_2}(i).$
\end{enumerate}

In fact, the Shapley value is the \emph{unique} value satisfying these four axioms. 
\begin{theorem} \label{thm_unique}
The Shapley value is the unique value satisfying Axioms 1--4 (\citet{Shapley1953, Winter2002}).
\end{theorem}

Note that the formulation described here is slightly different from the original formulation in \citet{Shapley1953}.  On the one hand, \citet{Shapley1953} used a framework consisting of three axioms:  symmetry, additivity, and a carrier axiom,  the latter comprising both Pareto efficiency and the dummy axiom (see \citet{Winter2002} for details). On the other hand, \citet{Shapley1953} made the additional assumption that $\nu$ is a superadditive function (i.e. $\nu(A \cup B) \geq \nu(A) + \nu(B)$ for all pairs of disjoint sets $A,B$), which was later relaxed by \citet{Dubey1975}. \\

In the phylogenetic setting, $\nu(S)$ is taken to be the {\em phylogenetic diversity of $S$ on $T$}\footnote{Note that PD is not a superadditive function. In fact, it is submodular, satisfying the property that $\nu(A \cup B)\leq \nu(A)+\nu(B)-\nu(A \cap B)$ for all $A, B$ (cf. Proposition 6.13 in \citet{Book_Mike}).}, denoted \blue{by} $PD_{\blue{T}}(S)$, and defined as the sum of lengths of the edges in the minimal subtree of $T$ that contains $S$ and the root of $T$ (cf. \citet{Faith1992}). As an example, for the tree $T'$ depicted in Fig. \ref{fig1}(b), and the subset $S=\{1,2,4\}$ of leaves, we have $PD_{\blue{T'}}(S)=11$. 

Considering the leaf set $[n]$ of a rooted phylogenetic tree $T$ as the set of players and phylogenetic diversity as the characteristic function of a game, Eqn. \eqref{def_sv} becomes:

$$
SV_T(i) = \frac{1}{n!} \sum\limits_{S \subseteq [n]: \, i \in S} (|S|-1)!  (n-|S|)!  \left(PD_{\blue{T}}(S) - PD_{\blue{T}}(S \setminus \{i\} )\right).
$$
Note that in contrast to the previous two sections we are not assuming in this section that $T$ is a binary tree.

In an important paper, \citet{fuc} proved that the Shapley value and the Fair Proportion index on rooted phylogenetic trees agree (see also \citet{Book_Mike} and \citet{sta}).

\begin{theorem}[\citet{fuc}] \label{thm_sv_fp}
The Fair Proportion index and the Shapley value are identical \blue{on rooted phylogenetic trees}, i.e. for all $i \in X$: $$FP_T(i) = SV_T(i).$$
\end{theorem}

In the following we will use this result to show that SV is the unique value satisfying Axioms 1--4 for the sub-class of games induced by a rooted tree $T$ and the phylogenetic diversity function. This is not obvious since \blue{(as noted by \cite{Haake2008} in the setting of PD on unrooted trees),}  the class of games \blue{based on PD on a rooted tree} is smaller than the class of all games (for which Theorem \ref{thm_unique} states that SV is unique). \blue{Apart from SV there might be other functions that satisfy these 4 axioms} for this smaller class of games, \blue{and so SV} might not be uniquely determined by them.   In Theorem \ref{thm_FP_unique}, however, we show that SV is still uniquely characterised by the 4 axioms for this smaller class of games.  \cite{Haake2008}, by contrast, 
introduced an additional axiom to obtain their characterization (Theorem 9 of that paper).

\blue{Let} $\mathcal{T}_{[n],PD_T}$ denote the class of games induced by a rooted phylogenetic tree $T$ with leaf set $[n]$ and non-negative edge lengths, and the phylogenetic diversity function on $T$.
Moreover, let a pair $([n], PD_T)$ denote a \emph{PD game}. Note that such a pair can be represented as a linear combination of so-called \emph{basis games} \blue{$PD_{T_e}$} (for $e \in E(T)$), where \blue{$PD_{T_e}$} corresponds to the $PD$ game on tree $T_e$, in which edge $e$ has length 1 and all other edges have \blue{length} 0. It can be shown that the family $\blue{(PD_{T_e})}_{e \in E(T)}$ is linearly independent and forms a basis of $\mathcal{T}_{[n],PD_T}$ of dimension $|E(T)|$.

The following theorem provides an axiomatic characterization of SV for games in $\mathcal{T}_{[n],PD_T}$.

\begin{theorem} \label{thm_FP_unique}
\blue{There is a unique function}  $$\blue{\psi_{PD_T}}: \mathcal{T}_{[n],PD_T} \rightarrow \mathbb{R}^n$$ that satisfies Axioms 1--4.  This  function coincides with the Shapley value, i.e. $\blue{\psi_{PD_T}}(i) = SV_T(i)$ for all $i \in [n]$.
\end{theorem}

\begin{proof}
By Theorem \ref{thm_unique}, SV satisfies Axioms 1--4.

Now, let $([n], PD_T)$ be a PD game and let $\blue{\psi_{PD_T}}$ satisfy all Axioms 1--4. We first consider a basis game \blue{$PD_{T_e}$} and determine $\blue{\psi_{PD_{T_e}}}$.

\blue{Let $N(e)$ denote the set of leaves descended from $e$ and let $n(e) = \vert N(e) \vert$. Then, all leaves not in $N(e)$ are dummy players, as for all $j \in [n] \setminus N(e)$, we have that $PD_{T_e}(C \cup \{j\}) = PD_{T_e}(C)$ for all $C \subseteq [n] \setminus \{j\}$. As $\psi_{PD_{T_e}}$ satisfies the dummy axiom, this implies that $\psi_{PD_{T_e}}(j)=0$ for all $j \in [n] \setminus N(e)$. \blue{On the other hand, all leaves in $N(e)$ are symmetric players as for any pair $i, j \in N(e)$ (with $i \neq j$), we have that $PD_{\blue{T_e}}(C\cup \{i\})=PD_{\blue{T_e}}(C\cup \{j\})=1$ holds for all subsets $C$ of $[n] \setminus \{i,j\}$.}}
As this holds for all pairs $i, j \in N(e)$ and as $\psi_{\blue{PD_{T_e}}}$ satisfies symmetry, we can conclude that $\psi_{\blue{PD_{T_e}}}(i) = \psi_{\blue{PD_{T_e}}}(j)$ for all \blue{$i \neq j \in N(e)$}. On the other hand, since $\psi_{\blue{PD_{T_e}}}$ satisfies efficiency, we have
\begin{align*}
PD_{\blue{T_e}}([n]) &= 1 = \sum\limits_{j \in [n]} \psi_{\blue{PD_{T_e}}}(j) \\
&=  \sum\limits_{j \in [n] \setminus N(e)} \underbrace{\psi_{\blue{PD_{T_e}}}(j)}_{=0} + \sum\limits_{i \in N(e)} \psi_{\blue{PD_{T_e}}}(i) =\sum\limits_{i \in N(e)} \psi_{\blue{PD_{T_e}}}(i),
\end{align*}
which -- using symmetry -- implies that $\psi_{\blue{PD_{T_e}}}(i) = \frac{1}{n(e)}$ for all $i \in N(e)$. To summarize, $\psi_{\blue{PD_{T_e}}}(i) = \frac{1}{n(e)}$  for all $i \in N(e)$ and  $\psi_{\blue{PD_{T_e}}}(j) = 0$ for all $j \in [n] \setminus N(e)$. It is easily verified that these values coincide with \blue{the FP index and thus with the SV (by Theorem \ref{thm_sv_fp}).}

\blue{Analogously, one can show that $\psi_{PD_T}$ is a linear function. As $\psi_{PD_T}$ satisfies Axiom 4, it is additive.  Moreover, for all $\lambda \in \mathbb{R}_{\geq 0}$, let $\lambda \cdot PD_{T_e}$ denote the PD game on tree $T_e$, in which edge $e$ has length $\lambda \cdot 1 = \lambda$ and all other edges have length 0. Then, using the same notation and reasoning as above, we have for all $i \in [n] \setminus N(e)$, $\psi_{\lambda PD_{T_e}}(i) = 0$, and for all $i \in N(e)$, $\psi_{\lambda PD_{T_e}}(i) = \lambda/n(e)$. Comparing this with $\psi_{PD_{T_e}}(i)$ from above, it is now easy to see that we have $\psi_{\lambda \,PD_{T_e}}(i) = \lambda \cdot \psi_{PD_{T_e}}(i)$ for all $\lambda \in \mathbb{R}_{\geq 0}$ and all $i \in [n]$. }
 
\blue{Together with the} additivity of $\psi$ and SV this implies that $\psi$ coincides with SV for all games in $\mathcal{T}_{[n],PD_T}$.
\end{proof}

\begin{remark}
{\rm Since SV is the unique index satisfying Pareto efficiency, symmetry, the dummy axiom and additivity for the class of games induced by a rooted tree and PD (by Theorem \ref{thm_FP_unique}), and since SV and FP agree for rooted trees (by Theorem \ref{thm_sv_fp}) and, in general, ES $\neq$ FP, it follows that ES must violate at least one of these four axioms. It can easily be checked that ES satisfies Pareto efficiency, additivity and the dummy axiom, but it may violate symmetry. An example is given in Figure \ref{fig6}, where we have $ES_T(1) \neq ES_T(3)$, even though $PD_{\blue{T}}(C \cup \{1\}) = PD_{\blue{T}}(C \cup \{3\})$ for all $C \subseteq [4]\setminus \{1,3\}$  
(on the other hand, $FP_T(1)=FP_T(3)$).}
\end{remark}

\begin{figure}[htbp]
	\centering
	\includegraphics[scale=0.25]{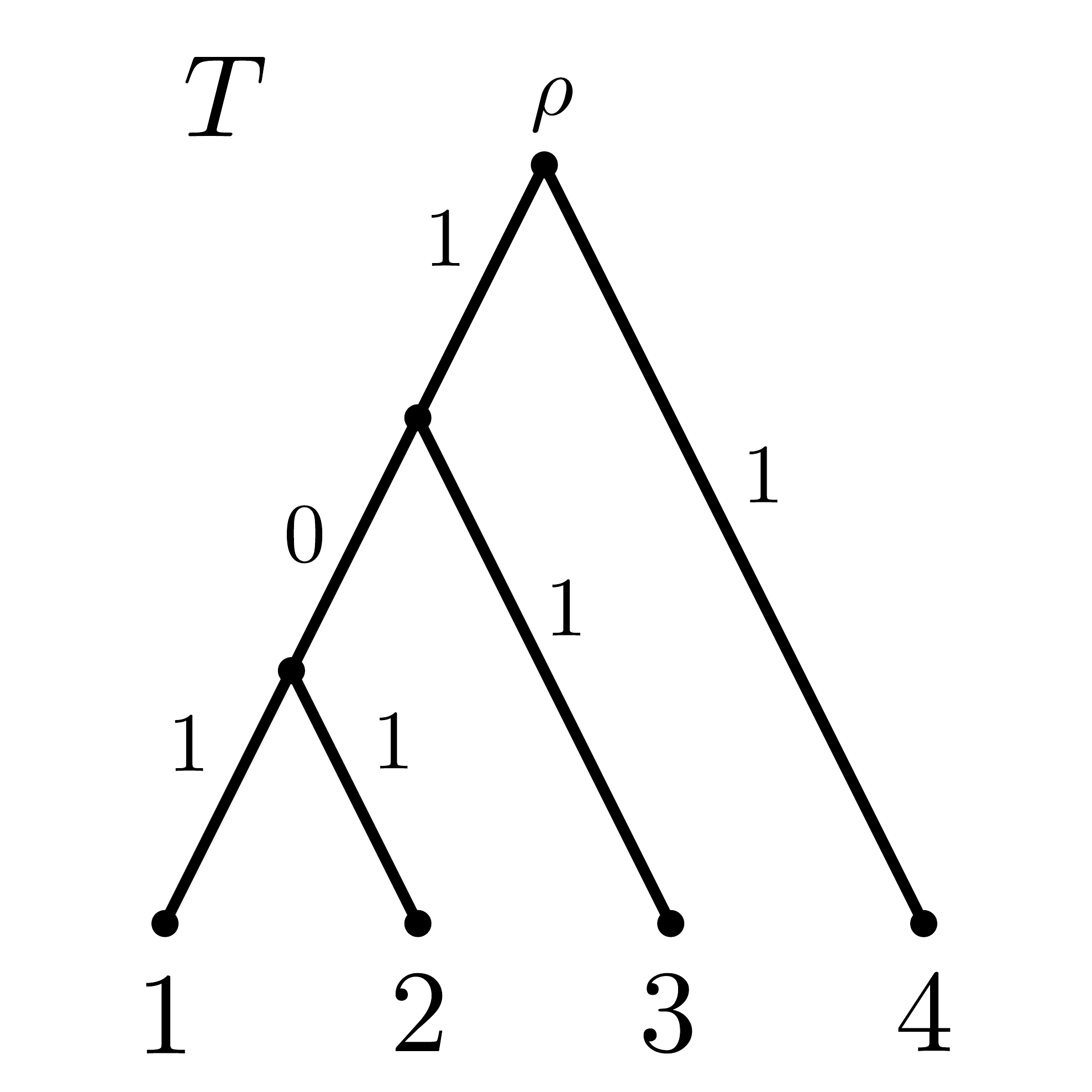}
	\caption{An example showing that ES fails to satisfy the symmetry axiom. Here, $ES_T(1)=5/4 \neq ES_T(3)=3/2$, even though $PD_{\blue{T}}(C \cup \{1\}) = PD_{\blue{T}}(C \cup \{3\})$ for all $C \subseteq [4] \setminus \{1,3\}$. By contrast, for  FP, we have $FP_T(1)=FP_T(3)=4/3$.}
	\label{fig6}
\end{figure}


\section{Diversity indices for unrooted trees}
\label{sec_un}

We now consider phylogenetic diversity indices for unrooted trees. An unrooted tree $T$ with leaf set $X$  is said to be an unrooted phylogenetic $X$--tree if each non-leaf vertex is unlabelled and has degree at least 3 (two such trees are considered equivalent if there is a
graph isomorphism between them that sends leaf $x$ to leaf $x$ for each $x\in X$).  In the case where all non-leaf vertices  in $T$ have degree exactly equal to 3, $T$ is said to be {\em binary}. Background on the basic combinatorics of unrooted phylogenetic trees can be found in  \cite{Book_Mike}.

Let $T$ be an unrooted phylogenetic tree (not necessarily binary) with leaf set $X=[n]$ and let all edges $e$ have non-negative edge lengths $l(e)$. For a subset $Y$ of the leaves, the (unrooted) phylogenetic diversity of $Y$ is defined as the sum of the edge lengths of the minimal subtree  connecting the leaves in $Y$. Note that $PD_{\blue{T}}(\{i\})=0$ for all $i \in [n]$  and $PD_{\blue{T}}([n])$  is the total sum of edge lengths of $T$ (i.e $\sum_e l(e)$). For a leaf $i \in [n]$ and an edge $e$ of $T$, let $I(T; i, e)$ be the set of interior vertices of $T$ in the path in $T$ from $i$ to edge $e$ (including the first vertex of $e$ that is reached, but not the second),  and for each vertex $v$ of $T$ let $d(v)$ denote the degree of $v$. 
For each edge $e$ of $T$, let 
\begin{equation}
\label{Def_mu}
\mu(i,e)= \frac{1}{2} \prod\limits_{v \in I(T; i,e)} \frac{1}{d(v)-1}, 
\end{equation}
 where we adopt the convention that if  $I(T; i, e) = \emptyset$ (i.e. $e$ is a pendant edge incident with leaf $i$) then $\prod\limits_{v \in I(T; i,e)} \frac{1}{d(v)-1}=1$ and hence
$\mu(i,e) =1/2$.

\subsection{Unrooted Equal Splits}
In this section, we develop a version of Equal Splits for unrooted trees. Recall that for rooted trees, the definition of the ES index is $ES_T(i) = \sum\limits_{e \in P(T; \rho, i)} \frac{1}{\Pi(e,i)}l(e)$, where $\Pi(e,i) = 1$ if $e$ is a pendant edge incident with $i$; otherwise, if $e=(u,v)$ is an interior edge, then $\Pi(e, i)$ is the product of the out-degrees of the interior vertices on the directed path 
from $v$ to leaf $i$.

This definition does not directly apply to unrooted trees, since there is no reference root vertex $\rho$ in an unrooted tree. Moreover, introducing  a phantom root vertex in an unrooted  tree results in different ES index values, depending on where the phantom root is inserted. Nevertheless, we can define a canonical unrooted version of ES that is a diversity index as follows.

Let $$\varphi_{\rm ES}(i)= \sum_{e} \mu_{\rm ES}(i,e) l(e),$$ where the summation is over all edges of $T$ and where 
$$\mu_{\rm ES}(i, e) =
\begin{cases}
1,  &  \mbox{ if $e$ is a pendant edge incident with $i$}; \\
\mu(i,e) , & \mbox{ if $e=\{u,v\}$ is an interior edge of $T$};\\
0,  & \mbox{ otherwise}. \\
\end{cases}
$$
Note that $\mu(i,e)$ is the expression introduced in Eqn. \eqref{Def_mu}. 
Moreover, note that in contrast to the rooted setting, $\varphi_{\rm ES}(i)$ is defined as a sum over \emph{all} edges of $T$ and not only over edges on a certain path in $T$. In fact, even though pendant edges not incident with leaf $i$ do not contribute to $\varphi_{\rm ES}(i)$ (since $\mu(i,e)=0$ in that case), the edges that do contribute do not necessarily form a path in $T$ (cf. Fig. \ref{fig7}).

\begin{theorem} \label{Thm_unrootedES}
For any unrooted phylogenetic tree $T$, $\varphi_{\rm ES}$ is a diversity index for $T$. In other words:
\begin{equation}
\label{unrooted_index}
\sum_{i\in [n]} \varphi_{\rm ES}(i)=PD_{\blue{T}}([n]).
\end{equation}
\end{theorem}

In order to prove this theorem, we require the following technical lemmas:

\begin{lemma}
\label{lemma_sm}
Suppose that $T$ is a rooted phylogenetic tree with leaf set $Y$ and root vertex $u$. Let $d^{-}(v)$ denote the out-degree of vertex $v$.
We then have:
$$\sum_{i \in Y} \prod_{v \in I(T; u, i)} \frac{1}{d^{-}(v)} = 1.$$
\end{lemma}
\begin{proof}
We use a simple probabilistic argument. Consider a random walk, starting from the root vertex $u$ and proceeding towards the leaves. At each interior vertex $v$, one of the $d^{-}(v)$
child vertices of $v$ is chosen uniformly at random (and independently of earlier choices).  In this way, the probability $p_i$ of arriving at leaf $i$ is simply $\prod_{v \in I(T; u, i)} \frac{1}{d^{-}(v)}$. Since we always arrive at one (and only one) leaf of $Y$ by this process, $\sum_{i \in Y} p_i=1$, as required.  
\end{proof} \hfill$\Box$ \\

\begin{corollary} \label{lemma_sum}
Let $T$ be an unrooted phylogenetic tree with leaf set $[n]$ and let $e=\{u,v\}$ be an arbitrary edge of $T$. Let $A$ and $B$ denote the subsets of leaves of $T$ that lie on each side of $e$, with $A$ being closer to $u$ (if $u$ is a leaf, then $A=\{u\}$) and $B$ being closer to $v$ (again, if $v$ is a leaf, then $B=\{v\}$). In this case:
$$ \sum\limits_{i \in A} \prod\limits_{v \in I(T; i, e)} \frac{1}{d(v)-1} = 1 \mbox{ and } \sum\limits_{i \in B} \prod\limits_{v \in I(T; i, e)} \frac{1}{d(v)-1} = 1. $$
\end{corollary}

\begin{proof}
Clearly, we only have to prove the first statement, so  consider $i \in A$.
If $|A|=1$ (which implies $i=u$), $I(T; i,e) = \emptyset$, and we again adopt the convention that in this case $\prod\limits_{v \in I(T;i,e)} \frac{1}{d(v)-1}=1$. 
In particular, the claimed statement holds for $|A|=1$. 
Next, consider $|A| > 1$. 
Then the expression 
$$ \sum\limits_{i \in A} \prod\limits_{v \in I(T; i, e)} \frac{1}{d(v)-1}$$
can also be written as:
\begin{equation}
\label{unrooted_to_rooted}
\sum_{i \in A}\prod_{v \in I(T_A; u,i)}\frac{1}{d^{-}(v)},
\end{equation}
where $T_A$ is the rooted phylogenetic tree on leaf set $A$ and root vertex $u$ obtained from $T$ by deleting edge $e$ and the subtree of $T$ with leaf set $B$.
The corollary now follows from Lemma~\ref{lemma_sm} by taking $T_A$ as the tree in that lemma, and $Y=A$. 
Note that the statement can alternatively be shown without the use of Lemma \ref{lemma_sm} by using an inductive argument.
\end{proof}
 \hfill$\Box$ \\

\begin{lemma} \label{lemma_eq}
Suppose that the linear equation 
\begin{equation}
\label{linear}
\sum_{e \in E} a(e)l(e) = \sum_{e\in E}  b(e)l(e),
\end{equation} 
with $a(e), b(e) \in \RR$, 
holds for all choices of $l$ of the form $l=l_{e'}$ where $e' \in E$ and 
$$l_{e'}(e)=\begin{cases}
1, & \mbox{ if $e=e'$};\\
0, & \mbox{ if $e \neq e'$}.
\end{cases}
$$
Eqn.~(\ref{linear}) then holds for all choices of $l \in \RR^E$.
\end{lemma}
\begin{proof}
The proof involves simple linear algebra.
Let $c(e):=a(e)-b(e)$.  Eqn.~(\ref{linear}) can then be rewritten as $\sum_{e \in E} c(e)l(e)=0$. Suppose this equation holds whenever $l=l_{e'}$ (and for each choice of $e'$).  
Then this equation becomes $c(e')\cdot 1 = 0$, and since this holds for all choices of $e'$, all the $c$--coefficients are zero, which gives the result.
\end{proof} \hfill$\Box$ \\
We are now in the position to prove Theorem \ref{Thm_unrootedES}.

\bigskip

{\em Proof of Theorem \ref{Thm_unrootedES}:}
By Lemma~\ref{lemma_eq}, it suffices to establish Eqn.~(\ref{unrooted_index}) when  $l$ assigns length 1 to an arbitrary edge $e'=\{u,v\}$ and 0 to all other edges.   Then $PD_{\blue{T}}([n]) =1$ and the left hand side of Eqn.~(\ref{unrooted_index}) is $\sum_{i \in [n]} \mu_{\rm ES}(i, e')$.
Our aim then is to show that this last quantity is always equal to $\sum_{i \in [n]} \mu_{\rm ES}(i,e')=1$.  This is true by definition of $\mu_{\rm ES}$ when $e'$ is a pendant edge, so we may suppose that $e'$ is an interior edge. In that case, let  $A$ and $B$ denote the subsets of leaves of $T$ that lie on each side of $e'$, with $A$ being closer to $u$ than $v$ and $B$ being closer to $v$ than $u$ (thus $A \cup B = [n]$, $A\cap B=\emptyset$ and 
$|A|, |B| \geq 2$). 
Since $\mu_{\rm ES}(i, e')=\mu(i, e')$ (since $e'$ is an interior edge) we have:
\begin{equation}
\label{twosum}
\sum_{i \in [n]}\mu(i, e') = \frac{1}{2} \left(\sum_{i \in A} 2 
\mu(i, e') + \sum_{i \in B} 2 \mu(i, e')\right),
\end{equation}
and
$$\sum_{i \in A} 2\mu(i, e') = \sum_{i \in A}\prod_{v \in I(T; i,e')}\frac{1}{d(v)-1}=1,$$
where the last equality follows from Corollary \ref{lemma_sum}.
A similar argument shows that
$\sum_{i \in B} 2\mu(i,e')=1$, and so, by Eqn.~(\ref{twosum}), we obtain the required equality:
$$\sum_{i \in [n]} \mu(i, e') = \frac{1}{2}(1+1)=1 = PD_{\blue{T}}([n]).$$ 
\hfill$\Box$

\subsection{A diversity index related to the Pauplin representation of phylogenetic diversity}
$PD_{\blue{T}}([n])$ can also be expressed as a positive linear combination of the pairwise distances $d(i,j) = \sum_{e \in P(T; i,j)} l(e)$ between leaves $i$ and $j$ in various ways, one of them being the following representation described by \citet{Semple2004}:

\begin{align}
PD_{\blue{T}}([n]) &=  \sum\limits_{\{i,j\} \subseteq [n]} \lambda_{ij} \, d(i,j) 
= \frac{1}{2} \sum\limits_{i=1}^n \, \sum\limits_{\substack{j=1 \\ j \neq i}}^n \, \lambda_{ij} \, d(i,j) 
= \sum\limits_{i=1}^n \, \frac{1}{2} \, \sum\limits_{\substack{j=1 \\ j \neq i}}^n \, \lambda_{ij} \, d(i,j), \label{Pauplin}
\end{align}
where $$\lambda_{ij} = \prod\limits_{v \in I(T; i,j)} \frac{1}{d(v)-1},$$ and where $I(T; i,j)$ denotes the set of interior vertices on the path from $i$ to $j$ in $T$.

Although this representation holds for general trees (not only binary ones), for binary trees, this expression is also known as the Pauplin representation of phylogenetic diversity (cf. \citet{Pauplin2000}). In the following section, we will further analyse this representation and suggest that it leads to yet another possible unrooted PD index.
Let $$\varphi_{\rm Pa}(i) = \sum\limits_{e} \mu(i, e) \; l(e),$$
where the summation is over all edges of $T$ and $\mu(i, e)$ is the expression introduced in Eqn. \eqref{Def_mu}.

\begin{theorem}
Let $T$ be an unrooted phylogenetic tree with leaf set $[n]$ and let $i$ be a leaf of $T$. In that case:
 \begin{equation}\label{edge_representation}
	\varphi_{\rm Pa}(i) = \frac{1}{2} \sum\limits_{\substack{j=1 \\ j \neq i}}^n \lambda_{ij} d(i,j).
	\end{equation} 
In other words, $\varphi_{\rm Pa}$ is closely related to the Pauplin representation of PD given in Eqn. \eqref{Pauplin}.
Moreover, $\varphi_{\rm Pa}$ is a diversity index (i.e. $\sum_{i \in [n]} \varphi_{\rm Pa}(i) = PD_{\blue{T}}([n])$).
\end{theorem}

\begin{proof}
Let $i \in [n]$ be a leaf of $T$. 
By Lemma~\ref{lemma_eq} it suffices to establish Eqn. \eqref{edge_representation} when  $l$ assigns length 1 to an arbitrary edge $e'=\{u,w\}$ and 0 to all other edges. 
Note that the removal of edge $e'$ splits $T$ into two subtrees. Let $C$ (=`close') denote the leaf set of the subtree that contains leaf $i$ and let $F$ (=`far') denote the leaf set of the other subtree. Now, for all leaves $j \neq i$ we clearly have:
	\begin{align*}
	d(i,j) &= \begin{cases}
		0, \text{ if } j \in C, \\
		1, \text{ if } j \in F.
		\end{cases}
	\end{align*}
Thus, we have for the right-hand side of Equation \eqref{edge_representation}
\begin{align*}
\frac{1}{2} \, \sum\limits_{\substack{j=1 \\ j \neq i}}^n \, \prod\limits_{v \in I(T; i,j)} \frac{1}{d(v)-1} \, d(i,j)
= \frac{1}{2} \, \sum\limits_{j \in F} \, \prod\limits_{v \in I(T; i,j)} \frac{1}{d(v)-1}.
\end{align*}
As $e'=\{u,w\}$ lies on the path from $i$ to $j$, the term on the right of this last equation can also be written as:
\begin{align*}
\frac{1}{2} \,  \sum\limits_{j \in F} \, \prod\limits_{v \in I(T; i, e')} \frac{1}{d(v)-1} \, \prod\limits_{v' \in I(T; j, e')} \frac{1}{d(v')-1}
&= \frac{1}{2} \, \prod\limits_{v \in I(T; i, e')} \frac{1}{d(v)-1} \, \sum\limits_{j \in F} \, \prod\limits_{v' \in I(T; j, e')} \frac{1} {d(v')-1} \\
&= \frac{1}{2} \, \prod\limits_{v \in I(T; i, e')} \frac{1}{d(v)-1}, 
\end{align*}
where the last equality follows from applying Corollary \ref{lemma_sum}.
On the other hand, for the left-hand side of Equation \eqref{edge_representation}, we have:
\begin{align*}
\varphi_{\rm Pa}(i) &= \frac{1}{2} \sum\limits_{e} \, \prod\limits_{v \in I(T; i, e)} \frac{1}{d(v)-1} \; l(e) 
= \frac{1}{2} \, \prod\limits_{v \in I(T; i, e')} \frac{1}{d(v)-1}, 
\end{align*}
as edge $e'=\{u,w\}$ has length 1, while all other edges have length 0, which completes the proof of Eqn.~(\ref{edge_representation}).
The claim that  $\varphi_{\rm Pa}$ is a diversity index is now a direct consequence from Eqn.~\eqref{Pauplin}.
\end{proof} \hfill$\Box$ \\

\subsection{Unrooted Fair Proportion}
Similar to the Equal Splits index, the Fair Proportion index has so far only been considered for rooted trees. 
In the following, we suggest two canonical extensions of Fair Proportion to unrooted trees.
Recall that for rooted trees, the definition of FP is $FP_T(i) = \sum\limits_{e \in P(T; \rho, i)} \frac{1}{n(e)} l(e)$, where $n(e)$ is the number of leaves descended from $e$. Note that the removal of edge $e$ splits $T$ into two connected components and $n(e)$ is the number of leaves of $T$ in the connected component that contains $i$. This concept can be extended to unrooted trees as follows.

For a leaf $i \in [n]$ and an edge $e$ of $T$, let $c(i,e)$ denote the size of the set of leaves that lie on the same side of $e$ as $i$. 
Let 
$$ \varphi_{\rm FP}(i) = \frac{1}{2} \sum\limits_{e} \frac{1}{c(i,e)} l(e),$$
and let
$$ \tilde{\varphi}_{\rm FP}(i) = \sum\limits_{e} \mu_{\rm FP}(i,e) l(e),$$
where the summation is over all edges of $T$ and where
$$ \mu_{\rm FP}(i, e) = 
\begin{cases}
1,  &  \mbox{ if $e$ is a pendant edge incident with $i$}; \\
\frac{1}{2 c(i,e)} , & \mbox{ if $e=\{u,v\}$ is an interior edge of $T$};\\
0,  & \mbox{ otherwise}. \\
\end{cases}$$

\begin{theorem} \label{Thm_Unrooted_FP}
For any unrooted phylogenetic tree $T$, $\varphi_{\rm FP}$ and $\tilde{\varphi}_{\rm FP}$ are diversity indices for $T$. In other words, 
\begin{align}
\sum\limits_{i \in [n]} \varphi_{\rm FP}(i) &= PD_{\blue{T}}([n])  \mbox{ and } \label{FP_unrooted} \\
\sum\limits_{i \in [n]} \tilde{\varphi}_{\rm FP}(i) &= PD_{\blue{T}}([n]). \label{TildeFP_unrooted}
\end{align}
\end{theorem}

\begin{proof}
We first establish Eqn. \eqref{FP_unrooted}. 
By Lemma~\ref{lemma_eq}, it suffices to establish Eqn.~(\ref{FP_unrooted}) when  $l$ assigns length 1 to an arbitrary edge $e'=\{u,v\}$ and 0 to all other edges.  Then, $PD_{\blue{T}}([n])=1$ and the left hand side of Eqn. \eqref{FP_unrooted} is $\frac{1}{2} \sum_{i \in [n]} \frac{1}{c(i,e')}$. Now, let $A$ and $B$ denote the subsets of leaves that lie on each side of $e'$ (i.e. $A \cup B = [n]$, $A \cap B = \emptyset$ and $|A|, \, |B| \geq 1$), in which case:
\begin{align*}
\frac{1}{2} \sum_{i \in [n]} \frac{1}{c(i,e')} 
= \frac{1}{2} \left( \sum_{i \in A}  \frac{1}{c(i,e')} + \sum_{i \in B} \frac{1}{c(i,e')} \right)
= \frac{1}{2} \left( \sum_{i \in A}  \frac{1}{|A|} + \sum_{i \in B} \frac{1}{|B|} \right)\\
= \frac{1}{2} \left( 1 + 1 \right) 
= 1 = PD_{\blue{T}}([n]).
\end{align*}
Eqn. \eqref{TildeFP_unrooted} follows from a similar argument by noting that the left hand side of this equation becomes $ \sum_{i \in [n]} \mu_{\rm FP}(i,e')$. If $e'$ is a pendant edge, this quantity is equal to 1 by definition of $\mu_{\rm FP}$ and if $e'$ is an interior edge, the same reasoning as in the proof of Eqn. \eqref{FP_unrooted} establishes $\sum_{i \in [n]} \mu_{\rm FP}(i,e') = 1 = PD_{\blue{T}}([n])$.
\end{proof} \hfill$\Box$ \\

\subsection{Summary of unrooted diversity indices}
In the last sections we have presented canonical extensions of Equal Splits and Fair Proportion to unrooted trees and have also introduced a diversity index closely related to the Pauplin representation of phylogenetic diversity. 
Although  all these indices appear to be new, an unrooted Shapley value has long been known in the literature. In fact, even though the Shapley value is frequently used for rooted trees, it was first defined and introduced for unrooted trees by \citet{Haake2008} and can be expressed as follows:
$$\varphi_{\rm SV}(i) = \sum_{e} \frac{f(i,e)}{n \, c(i,e)} l(e),$$
where the summation is over all edges of $T$, $c(i,e)$ is again the number of leaves that lie on the same side of $e$ as leaf $i$, and $f(i,e)$ is the number of leaves that lie on the other side of $e$ (cf. Theorem 4 in \citet{Haake2008}). Recall that for rooted trees, FP and SV are equivalent, so one might argue that the unrooted SV can be considered an unrooted analogue of FP. It turns out, however, that there exists a natural extension of FP to unrooted trees, that is different from unrooted SV. 

In fact, although all of the unrooted diversity indices discussed above can be expressed as linear functions of the edge lengths $l(e)$ of $T$ with coefficients that are independent of $l$, these coefficients differ among indices (cf. Table \ref{Table_Coefficients}) and the indices are, in general, not equivalent (cf. Figure \ref{fig7}).

\begin{figure}
	\centering
	\includegraphics[scale=1]{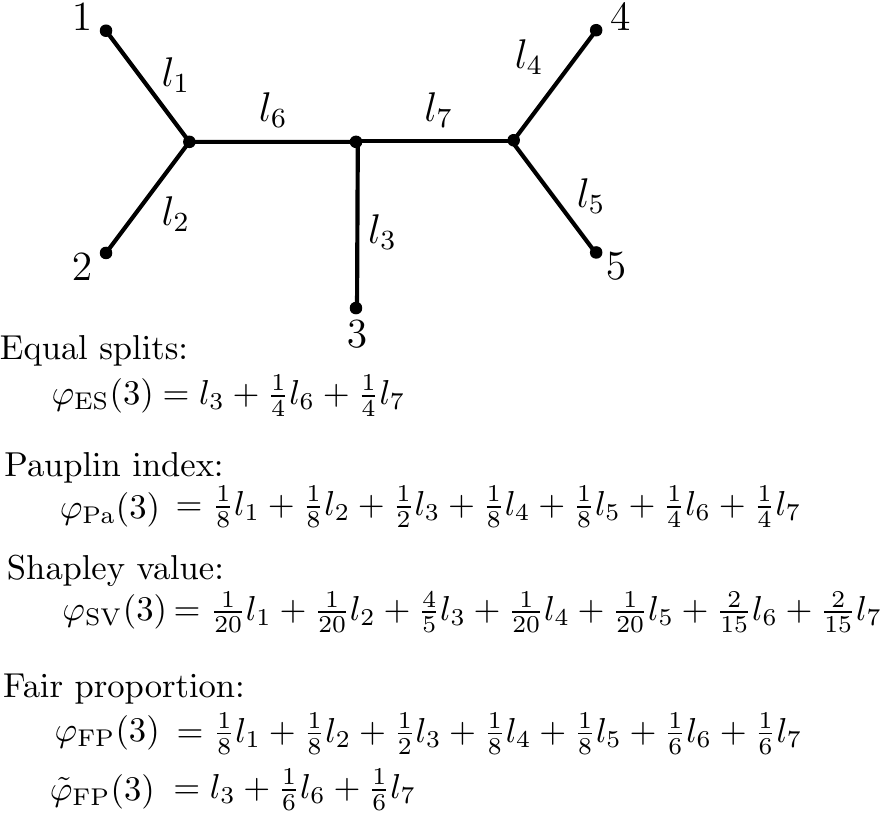}
	\caption{Unrooted binary phylogenetic tree on leaf set $[5]$ and the various unrooted PD indices for leaf $3$.}
	\label{fig7}
\end{figure}

\begin{table}[htbp]
\caption{Coefficients $\gamma_T(i,e)$ used in the calculation of $\varphi(i) = \sum_e \gamma_T(i,e) l(e)$, where $i$ is a leaf and $e$ is an edge of $T$. Moreover, $\mu(i,e)$ is as in Eqn. \eqref{Def_mu}, $c(i,e)$ denotes the number of leaves on the same side of $e$ as leaf $i$ and $f(i,e)$ denotes the number of leaves on the other side of $e$.} 
\label{Table_Coefficients}
\centering
\setlength{\tabcolsep}{2mm}
\renewcommand{\arraystretch}{1.5}
\begin{tabular}{lccc}
\toprule
\hspace*{5mm} & $e$ pendant edge incident with $i$ & $e$ interior edge & $e$ pendant edge not incident with $i$ \\
\midrule
$\varphi_{\rm ES}$ & $1$ & $\mu(i, e)$ & $0$ \\
$\varphi_{\rm Pa}$ & $\mu(i, e)$ & $\mu(i, e)$ & $\mu(i, e)$ \\
$\varphi_{\rm FP}$ & $\frac{1}{2 \, c(i,e)}$ &  $\frac{1}{2 \, c(i,e)}$ &  $\frac{1}{2 \, c(i,e)}$  \\
$\tilde{\varphi}_{\rm FP}$ & $1$ &  $\frac{1}{2 \, c(i,e)}$ &  $0$  \\
$\varphi_{\rm SV}$ & $\frac{f(i,e)}{n \, c(i,e)}$ & $\frac{f(i,e)}{n \, c(i,e)}$ & $\frac{f(i,e)}{n \, c(i,e)}$ \\
\bottomrule
\end{tabular}
\end{table}


\section{Concluding Remarks}

Phylogenetic diversity indices play a key role in biodiversity, so it is helpful to understand how the different indices are related. In this paper, we asked just how different they can be for rooted trees (in an extreme  sense, rather than on average). We also considered how some of the natural indices in the rooted settings extend to the unrooted setting, and further explored  the way in which the Shapley value relates to rooted and unrooted indices.  
Our work suggests two broad questions that may be interesting to explore in future work. First, how do the results in Sections~\ref{sec_diff} and \ref{sec3} extend if we lift the assumption that the underlying trees are binary? Second, for the unrooted indices in Section~\ref{sec_un}, how different can they be from one another (in the sense we considered in Section~\ref{sec_diff}) and for which trees are certain indices identical (in the sense we considered in Section \ref{sec3})?
Moreover, as all unrooted indices apart from the unrooted SV appear to be new, it additionally might be of interest to analyse their biological interpretation and relevance for conservation decisions.

\section{Acknowledgements}
We thank Arne Mooers for a number of helpful suggestions, and the two anonymous reviewers for detailed comments on an earlier version of this manuscript.
We also thank Fran{\c{c}}ois Bienvenu for pointing out an alternative proof of Lemma \ref{Lemma_Convergence}, and Mareike Fischer for helpful comments concerning Section 4. 
The first author also thanks the German Academic Scholarship Foundation for a doctoral scholarship.

\bibliographystyle{model2-names}
\bibliography{wickesteel_final_version} 

\begin{thebibliography}{17}
\expandafter\ifx\csname natexlab\endcsname\relax\def\natexlab#1{#1}\fi
\providecommand{\url}[1]{\texttt{#1}}
\providecommand{\href}[2]{#2}
\providecommand{\path}[1]{#1}
\providecommand{\DOIprefix}{doi:}
\providecommand{\ArXivprefix}{arXiv:}
\providecommand{\URLprefix}{URL: }
\providecommand{\Pubmedprefix}{pmid:}
\providecommand{\doi}[1]{\href{http://dx.doi.org/#1}{\path{#1}}}
\providecommand{\Pubmed}[1]{\href{pmid:#1}{\path{#1}}}
\providecommand{\bibinfo}[2]{#2}
\ifx\xfnm\relax \def\xfnm[#1]{\unskip,\space#1}\fi
\bibitem[{Dubey(1975)}]{Dubey1975}
\bibinfo{author}{Dubey, P.}, \bibinfo{year}{1975}.
\newblock \bibinfo{title}{On the uniqueness of the {S}hapley value}.
\newblock \bibinfo{journal}{International Journal of Game Theory}
  \bibinfo{volume}{4}, \bibinfo{pages}{131--139}.
\newblock \URLprefix \url{https://doi.org/10.1007/BF01780630},
  \DOIprefix\doi{10.1007/BF01780630}.
\bibitem[{Faith(1992)}]{Faith1992}
\bibinfo{author}{Faith, D.P.}, \bibinfo{year}{1992}.
\newblock \bibinfo{title}{Conservation evaluation and phylogenetic diversity}.
\newblock \bibinfo{journal}{Biological Conservation} \bibinfo{volume}{61},
  \bibinfo{pages}{1--10}.
\newblock \URLprefix \url{http://dx.doi.org/10.1016/0006-3207(92)91201-3},
  \DOIprefix\doi{10.1016/0006-3207(92)91201-3}.
\bibitem[{Fuchs and Jin(2015)}]{fuc}
\bibinfo{author}{Fuchs, M.}, \bibinfo{author}{Jin, E.Y.}, \bibinfo{year}{2015}.
\newblock \bibinfo{title}{Equality of {S}hapley value and fair proportion index
  in phylogenetic trees}.
\newblock \bibinfo{journal}{Journal of Mathematical Biology}
  \bibinfo{volume}{71}, \bibinfo{pages}{1133--1147}.
\bibitem[{Haake et~al.(2008)Haake, Kashiwada and Su}]{Haake2008}
\bibinfo{author}{Haake, C.J.}, \bibinfo{author}{Kashiwada, A.},
  \bibinfo{author}{Su, F.E.}, \bibinfo{year}{2008}.
\newblock \bibinfo{title}{The {S}hapley value of phylogenetic trees}.
\newblock \bibinfo{journal}{Journal of Mathematical Biology}
  \bibinfo{volume}{56}, \bibinfo{pages}{479--497}.
\newblock \URLprefix \url{http://dx.doi.org/10.1007/s00285-007-0126-2},
  \DOIprefix\doi{10.1007/s00285-007-0126-2}.
\bibitem[{Isaac et~al.(2007)Isaac, Turvey, Collen, Waterman and
  Baillie}]{isa07}
\bibinfo{author}{Isaac, N.}, \bibinfo{author}{Turvey, S.T.},
  \bibinfo{author}{Collen, B.}, \bibinfo{author}{Waterman, C.},
  \bibinfo{author}{Baillie, J.}, \bibinfo{year}{2007}.
\newblock \bibinfo{title}{Mammals on the {EDGE}: {C}onservation priorities
  based on threat and phylogeny}.
\newblock \bibinfo{journal}{PLoS One} \bibinfo{volume}{2},
  \bibinfo{pages}{e296}.
\bibitem[{Pauplin(2000)}]{Pauplin2000}
\bibinfo{author}{Pauplin, Y.}, \bibinfo{year}{2000}.
\newblock \bibinfo{title}{Direct calculation of a tree length using a distance
  matrix}.
\newblock \bibinfo{journal}{Journal of Molecular Evolution}
  \bibinfo{volume}{51}, \bibinfo{pages}{41--47}.
\newblock \DOIprefix\doi{10.1007/s002390010065}.
\bibitem[{Redding(2003)}]{Redding2003}
\bibinfo{author}{Redding, D.W.}, \bibinfo{year}{2003}.
\newblock \bibinfo{title}{Incorporating genetic distinctness and reserve
  occupancy into a conservation priorisation approach}.
\newblock Master's thesis. University Of East Anglia, Norwich, UK.
\bibitem[{Redding et~al.(2008)Redding, Hartmann, Mimoto, Bokal, DeVos and
  Mooers}]{Redding2008}
\bibinfo{author}{Redding, D.W.}, \bibinfo{author}{Hartmann, K.},
  \bibinfo{author}{Mimoto, A.}, \bibinfo{author}{Bokal, D.},
  \bibinfo{author}{DeVos, M.}, \bibinfo{author}{Mooers, A.{\O}.},
  \bibinfo{year}{2008}.
\newblock \bibinfo{title}{Evolutionarily distinctive species often capture more
  phylogenetic diversity than expected}.
\newblock \bibinfo{journal}{Journal of Theoretical Biology}
  \bibinfo{volume}{251}, \bibinfo{pages}{606--615}.
\newblock \DOIprefix\doi{10.1016/j.jtbi.2007.12.006}.
\bibitem[{Redding et~al.(2014)Redding, Mazel and Mooers}]{Redding2014}
\bibinfo{author}{Redding, D.W.}, \bibinfo{author}{Mazel, F.},
  \bibinfo{author}{Mooers, A.O.}, \bibinfo{year}{2014}.
\newblock \bibinfo{title}{Measuring evolutionary isolation for conservation}.
\newblock \bibinfo{journal}{{PL}o{S} {O}ne} \bibinfo{volume}{9},
  \bibinfo{pages}{1--15}.
\newblock \URLprefix \url{https://doi.org/10.1371/journal.pone.0113490},
  \DOIprefix\doi{10.1371/journal.pone.0113490}.
\bibitem[{Redding and Mooers(2006)}]{Redding2006}
\bibinfo{author}{Redding, D.W.}, \bibinfo{author}{Mooers, A.{\O}.},
  \bibinfo{year}{2006}.
\newblock \bibinfo{title}{Incorporating evolutionary measures into conservation
  prioritization}.
\newblock \bibinfo{journal}{Conservation Biology} \bibinfo{volume}{20},
  \bibinfo{pages}{1670--1678}.
\newblock \DOIprefix\doi{10.1111/j.1523-1739.2006.00555.x}.
\bibitem[{Semple and Steel(2004)}]{Semple2004}
\bibinfo{author}{Semple, C.}, \bibinfo{author}{Steel, M.},
  \bibinfo{year}{2004}.
\newblock \bibinfo{title}{Cyclic permutations and evolutionary trees}.
\newblock \bibinfo{journal}{Advances in Applied Mathematics}
  \bibinfo{volume}{32}, \bibinfo{pages}{669--680}.
\newblock \DOIprefix\doi{https://doi.org/10.1016/S0196-8858(03)00098-8}.
\bibitem[{Shapley(1953)}]{Shapley1953}
\bibinfo{author}{Shapley, L.S.}, \bibinfo{year}{1953}.
\newblock \bibinfo{title}{A value for $n$--person games}, in:
  \bibinfo{booktitle}{Contributions to the Theory of Games ({AM}-28), Volume
  {II}}. \bibinfo{publisher}{Princeton University Press,},
  \bibinfo{address}{Princeton}, pp. \bibinfo{pages}{307--317}.
\newblock \DOIprefix\doi{10.1515/9781400881970-018}.
\bibitem[{Stahn(2017)}]{sta}
\bibinfo{author}{Stahn, H.}, \bibinfo{year}{2017}.
\newblock \bibinfo{title}{Biodiversity, {S}hapley value and phylogenetic trees:
  {S}ome remarks}.
\newblock \bibinfo{type}{WP2017- Nr 41}. AMSE. \bibinfo{address}{URL:
  https://halshs.archives-ouvertes.fr/halshs-01630069/document}.
\bibitem[{Steel(2016)}]{Book_Mike}
\bibinfo{author}{Steel, M.}, \bibinfo{year}{2016}.
\newblock \bibinfo{title}{Phylogeny: Discrete and random processes in
  evolution}.
\newblock \bibinfo{publisher}{Society for Industrial and Applied Mathematics},
  \bibinfo{address}{Philadelphia PA}.
\bibitem[{Steele(2004)}]{stee04}
\bibinfo{author}{Steele, J.M.}, \bibinfo{year}{2004}.
\newblock \bibinfo{title}{The Cauchy-Schwarz Master Class}.
\newblock \bibinfo{publisher}{Cambridge University Press},
  \bibinfo{address}{New York}.
\bibitem[{Vellend et~al.(2011)Vellend, K.Cornwell, Magnuson-Ford and
  O.Mooers}]{Vellend2011}
\bibinfo{author}{Vellend, M.}, \bibinfo{author}{K.Cornwell, W.},
  \bibinfo{author}{Magnuson-Ford, K.}, \bibinfo{author}{O.Mooers, A.},
  \bibinfo{year}{2011}.
\newblock \bibinfo{title}{Measuring {P}hylogenetic biodiversity}, in:
  \bibinfo{booktitle}{Biological diversity: Frontiers in measurement and
  assessment}. \bibinfo{publisher}{Oxford University Press},
  \bibinfo{address}{Oxford}. chapter~\bibinfo{chapter}{14}, pp.
  \bibinfo{pages}{194--207}.
\bibitem[{Winter(2002)}]{Winter2002}
\bibinfo{author}{Winter, E.}, \bibinfo{year}{2002}.
\newblock \bibinfo{title}{The {S}hapley value}, in: \bibinfo{editor}{Aumann,
  R.}, \bibinfo{editor}{Hart, S.} (Eds.), \bibinfo{booktitle}{Handbook of game
  theory with economic applications}. \bibinfo{edition}{1} ed..
  \bibinfo{publisher}{Elsevier}. volume~\bibinfo{volume}{3}.
  chapter~\bibinfo{chapter}{53}, pp. \bibinfo{pages}{2025--2054}.
\newblock \URLprefix \url{https://EconPapers.repec.org/RePEc:eee:gamchp:3-53}.

\end{thebibliography}

\section*{Appendix: Proof of Lemma~\ref{Lemma_Convergence}}

{\em Proof of Lemma~\ref{Lemma_Convergence}:}
We first establish the following identity by application of the `fundamental theorem of calculus'.  
Let $f:[0, h] \rightarrow [0,1]$ be any continuous function and let $c>0$.   
We then have:
\begin{equation}
\label{cal1}
\int_1^h f(x) \cdot e^{-c\int_0^x f(t) dt} dx = \frac{1}{c}\left(\exp(-c 
\int_0^1 f(t)dt) -\exp(-c \int_0^h f(t)dt)\right).
\end{equation}
To establish (\ref{cal1}), let $G(x) =  \exp(-c\int_0^x f(t) dt)$. 
Since $f$ is continuous, $G'(x) =-cf(x)G(x)$, so the left-hand side of Eqn.~(\ref{cal1}) can be written as $\frac{-1}{c} \int_1^h G'(x)dx = \frac{1}{c} (G(1)-G(h)),$ which gives Eqn.~(\ref{cal1}).

Now, for all $x \geq 1$, $\int_0^{x-1} f(t) dt \geq \int_0^{x} f(t) dt - 1$, since $f$ takes values in the interval $[0,1]$,  and thus (\ref{cal1}) gives:
$$
\int_1^h f(x) \cdot e^{-c\int_0^{x-1} f(t) dt} dx \leq 
\frac{e^c}{c}e^{-c\int_0^1 f(t)dt}.
$$
Taking $c=\ln(2)$ in this last inequality gives:
\begin{equation}
\label{cal3}
\int_1^h f(x)\cdot 2^{-\int_0^{x-1} f(t) dt} dx \leq  \frac{2}{\ln 2} 
2^{-\int_0^1 f(t)dt}.
\end{equation}
Let $g$ be a piecewise continuous function that takes the value $x_i$ on the open interval $(i, i+1)$, for each $i=0,1,\ldots, h-1$, and let $f_j,j \geq 1$, be a sequence of continuous functions that converges in the $L^2$ norm to $g$ (e.g. by Fourier series).  As $j \rightarrow\infty$, $\int_1^h f_j(x) \cdot e^{-c\int_0^{x-1} f_j(t) dt} dx$ then converges to
$\sum_{i=1}^h x_i 2^{-\sum_{j<i} x_i}$ and 
$ \frac{2}{\ln 2}\cdot 2^{-\int_0^1 f_j(t)dt}$ converges to $\frac{2}{\ln 2}\cdot2^{-x_0}$.  Inequality~(\ref{cal3}) now establishes the lemma.
\hfill$\Box$
\end{document}